\documentclass[12pt]{article}
\usepackage[round]{natbib}
\usepackage{amscd}
\usepackage{enumitem}
\usepackage{fullpage}
\usepackage{amsthm}
\usepackage{bbm}
\usepackage{mathrsfs}
\usepackage{amsmath}
\usepackage{alltt, amssymb}
\input xy 
\xyoption{all} 

\def\P {\ensuremath{\mathbb{P}}}

\def\C {\ensuremath{\mathbb{C}}}
\def\Q {\ensuremath{\mathbb{Q}}}
\def\L {\ensuremath{\mathbb{L}}}

\def\M {\ensuremath{\mathbb{M}}}

\def\R {\ensuremath{\mathbb{R}}}
\def\Z {\ensuremath{\mathbb{Z}}}
\def\F {\ensuremath{\mathbb{F}}}

\def\bA {\ensuremath{\mathbf{A}}}

\def\K {\ensuremath{\mathbb{K}}}

\def\U {\ensuremath{\mathbf{U}}}
\def\T {\ensuremath{\mathbf{T}}}
\def\t {\ensuremath{\mathbf{t}}}
\def\i {\ensuremath{\mathbf{i}}}
\def\y {\ensuremath{\mathbf{y}}}
\def\Y {\ensuremath{\mathbf{Y}}}
\def\X {\ensuremath{\mathbf{X}}}

\def\V {\ensuremath{\mathscr{V}}}
\def\IP {\ensuremath{\mathscr{I}}}
\def\J {\ensuremath{\mathscr{J}}}
\def\KP {\ensuremath{\mathscr{L}}}
\def\WP {\ensuremath{\mathscr{Z}}}
\def\VS {\ensuremath{\mathscr{V}^\star}}
\def\W {\ensuremath{\mathscr{W}}}
\def\WS {\ensuremath{\mathscr{W}^\star}}
\def\IS {\ensuremath{\mathscr{I}^\star}}
\def\JS {\ensuremath{\mathscr{J}^\star}}
\newcommand{\qu}{\mathbb{Q}}

\newcommand{\Cho}{\mathscr{C}^\star}
\newcommand{\ChoV}{\mathscr{C}}
\newcommand{\Chov}{c}

\newcommand{\num}{\widetilde{T}}
\newcommand{\res}{\mathrm{res}}
\newcommand{\corr}{\mathsf{H}}
\newcommand{\I}{\mathsf{I}}
\newcommand{\G}{\mathsf{G}}

\newtheorem{Ass}{Assumption}

\newtheorem{Theo}{Theorem}

\newtheorem{Prop}{Proposition}
\newtheorem{Lemma}{Lemma}

\title{Bit-size estimates for triangular sets in positive dimension}

\author{Xavier Dahan\\
Faculty of Mathematics, Ky\^ush\^u University\\
dahan@math.kyushu-u.ac.jp\\[2mm]
Abdulilah Kadri\\
Mathematics Department, The University of Western Ontario\\
akadri4@uwo.ca\\[2mm]
\'Eric Schost\\
Computer Science Department, The University of Western Ontario\\
eschost@uwo.ca}

\begin{document}

\maketitle

\begin{abstract}
  We give bit-size estimates for the coefficients appearing in
  triangular sets describing positive-dimensional algebraic sets
  defined over $\Q$. These estimates are worst case upper bounds; they
  depend only on the degree and height of the underlying algebraic
  sets. We illustrate the use of these results in the context of a
  modular algorithm.

  This extends results by the first and last author, which were
  confined to the case of dimension 0. Our strategy is to get back to
  dimension 0 by evaluation and interpolation techniques. Even though
  the main tool (height theory) remains the same, new difficulties
  arise to control the growth of the coefficients during the
  interpolation process. 
\end{abstract}

\paragraph{Keywords:} triangular set, regular chain, Chow form, height function, bit-size


\section{Introduction}

It is well known that for algorithms for multivariate polynomials with
rational coefficients, or involving parameters, small inputs can
generate very large outputs. We will be concerned here with the
occurrence of this phenomenon for the solution of polynomial systems.

To circumvent this issue, a natural solution is to find smaller
outputs. In dimension 0, if a parametrization of the solutions is
required through a ``Shape Lemma'' output, the Rational Univariate
Representation~\citep{AlBeRoWo94,Rouillier99}, or Kronecker
representation~\citep{GiLeSa01}, is usually seen to have smaller
coefficients than a lexicographic Gr{\"o}bner basis. It is obtained by
multiplying the Gr{\"o}bner basis elements by a well-chosen
polynomial. It turns out that if a ``triangular'' representation is
wanted, a similar trick can be employed, which, in most practical
situations, reduces the coefficients size.

While such experimental observations can drive the choice or the
discovery of a good data structure, it is desirable to dispose of a
theoretical argument to validate its efficiency. Bit-size estimates,
like the ones provided in this article for positive dimensional
situations, provide this kind of theoretical argument. A second use of
this kind of result, which will be illustrated later on, is to help
quantify success probabilities of some probabilistic modular
algorithms.

\paragraph{Triangular representations.}
Let $k$ be a field; all fields will have characteristic 0 in this
paper. For the moment, let us consider a 0-dimensional algebraic set
$V \subset \overline k^n$, defined over $k$, and let $I \subset
k[\X]=k[X_1,\dots,X_n]$ be its defining ideal. Our typical assumption
will be the following.

\begin{Ass} \label{ass1} For the lexicographic order $X_1 < \cdots <
  X_n$, the reduced Gr{\"o}bner basis of the ideal $I$ has the form
$$\left | \begin{array}{l} T_n(X_1,\dots,X_n) \\ \hspace{0.7cm} \vdots
\\ T_2(X_1,X_2) \\ T_1(X_1),
\end{array}\right .$$
where for $\ell \leq n$, $T_\ell$ depends only on $X_1,\dots,X_\ell$
and is monic in $X_\ell$.
\end{Ass}
Following~\citet{Lazard92}, we say that the polynomials
$(T_1,\dots,T_n)$ form a {\em monic triangular set}, or simply a
triangular set. This representation is well-suited to many problems
(see some examples in~\citep{Lazard92,AuVa00,Schost03,Schost03a}), as
meaningful information is easily read off on it. 

Several algorithmic and complexity questions remain open for this data
structure: this paper studies one of them.  For $V$ as in
Assumption~\ref{ass1}, we are interested in the ``space complexity''
of the representation of $V$ by means of $(T_1,\dots,T_n)$.  For $\ell
\leq n$, let $d_\ell$ be the degree of $T_\ell$ in $X_\ell$ and let
$V_\ell \subset\overline k^\ell$ be the image of $V$ by the projection
$(x_1,\ldots,x_n) \mapsto (x_1,\ldots,x_\ell)$; then, $d_1 \cdots
d_\ell$ is the cardinality of $V_\ell$.

Representing $T_\ell$ amounts to specifying at most $d_1\cdots d_\ell$
elements of $k$. If $k$ bears no particular structure, we cannot say
more in terms of the space complexity of such a representation. New
questions arise when $k$ is endowed with a notion of ``size'': then,
the natural question is to relate the size of the coefficients in
$T_\ell$ to quantities associated to $V_\ell$.

This kind of information is useful in its own sake, but is also
crucial in the development of algorithms to compute triangular
sets~\citep{Schost03,DaMoScWuXi05,DaJiMoSc08}, using in particular
modular techniques. Several variants exist of such algorithms, most of
them being probabilistic: integers are reduced modulo one or several
random primes, and free variables are specialized at random values.
To analyze the running time or the error probability of these
algorithms, {\it a priori} bounds on the size of the coefficients of
$(T_1,\dots,T_n)$ are necessary (as is the case for modular algorithms
in general: already for linear algebra algorithms, or gcd computations,
bounds such as e.g. Hadamard's are crucial). An example of such an
application is given in the last section of this paper, in the context
of a modular algorithm for triangular decomposition.

\smallskip

The previous paper~\citep{DaSc04} gave such space complexity results
for the following cases:
\begin{itemize}
\item $k=\Q$, in which case we are concerned with the bit-size
  of coefficients;
\item $k=K(\Y)$, where $K$ is a field and $\Y=Y_1,\dots,Y_m$ are
  indeterminates; in this case we are concerned with the degrees in
  $\Y$ of the numerators and denominators of the coefficients.
\end{itemize}
These two cases cover many interesting concrete applications; the
latter is typically applied over $K=\F_p$. The goal of this paper is
to present an extension of these results to the last important case:
polynomials defined over $k=\Q(\Y)$. The second item above already
covers the degree-related aspects; what is missing is the study of
the bit-size of coefficients.

Unfortunately, the techniques of~\citet{DaSc04} are unable to provide
such information. Indeed, they rely on the study of an appropriate
family of absolute values on $k$, together with a suitable notion of
{\em height} for algebraic sets over $k$: for $k=\Q$, these are the
classical $p$-adic absolute values, plus the Archimedean one, and
height measures arithmetic complexity; for $k=K(\Y)$, there are the
absolute values associated to irreducible polynomials in $K[\Y]$, plus
the one associated to the total degree on $K[\Y]$; then, height is a
measure of geometric complexity.

Extending this approach to our case would require a family of absolute
values that captures the notion of bit-size on $\Q(\Y)$. Gauss' lemma
implies that $p$-adic absolute values do extend from $\Q$ to $\Q(\Y)$,
but the Archimedean one does not. As a result, concretely, it seems
unfeasible to re-apply the ideas of~\citet{DaSc04} here. A different
approach will be used, using evaluation and interpolation techniques.

Following~\citet{DaSc04}, it is fruitful to study not only the
polynomials $(T_1,\dots,T_n)$, but a related family of polynomials
written $(N_1,\dots,N_n)$ and defined as follows. Observe that for
$\ell \leq n$, $(T_1,\dots,T_\ell)$ form a reduced Gr{\"o}bner basis;
for a polynomial $A$ in $k[\X]$, $A \bmod \langle T_1,\ldots,T_\ell
\rangle$ denotes the normal form of $A$ modulo the Gr{\"o}bner basis
$(T_1,\ldots,T_\ell)$.  Let $D_1=1$ and $N_1=T_1$; for $2 \le \ell \le
n$, we define
\begin{eqnarray*}
  D_\ell & = & \prod_{1 \leq i \leq \ell-1} \frac{\partial T_i}{\partial X_i} \mod \langle T_1,\ldots,T_{\ell-1}\rangle,\\
  N_\ell  & = & D_\ell T_\ell  \mod \langle T_1,\ldots,T_{\ell-1}\rangle.
\end{eqnarray*}
\noindent Note that $D_\ell$ is in $k[X_1,\dots,X_{\ell-1}]$ and
$N_\ell$ in $k[X_1,\dots,X_{\ell-1},X_\ell]$, and that $D_\ell$ is the
leading coefficient of $N_\ell$ in $X_\ell$. Our reason to introduce
the polynomials $(N_1,\dots,N_n)$ is that they will feature much
better bounds than the polynomials $(T_1,\dots,T_n)$; we lose no
information, since the ideals $\langle T_1,\dots,T_n\rangle$ and
$\langle N_1,\dots,N_n\rangle$ coincide. Remark that the polynomials
$(N_1,\dots,N_n)$ are not monic, but the leading coefficient $D_\ell$
of $N_\ell$ is invertible modulo $\langle
N_1,\dots,N_{\ell-1}\rangle$: as such, $(N_1,\dots,N_n)$ form a {\em
  regular chain}~\citep{AuLaMo99}.

\paragraph{Main result.}
After this general introduction, our precise setup will be the
following.  Consider first the affine space of dimension $m+n$ over
$\C$, endowed with coordinates $\Y=Y_1,\dots,Y_m$ and
$\X=X_1,\dots,X_n$. For $0 \le \ell \le n$, let next $\Pi_\ell$ be the
projection
$$\begin{array}{cccc}
  \Pi_\ell: & \C^{m+n} & \to&  \C^{m+\ell}\\
  & (y_1,\dots,y_m,x_1,\dots,x_n) & \mapsto & (y_1,\dots,y_m,x_1,\dots,x_\ell),
\end{array}$$
so that $\Pi_0$ is the projection on the $\Y$-space.  Our starting
object will be a positive-dimensional algebraic set $\V$ defined over $\Q$;
then, the construction of the previous paragraphs will take place over
$k=\Q(\Y)$.

To measure the complexity of $\V$, we let $d_\V$ and $h_\V$ be
respectively its {\em degree} and {\em height}. For the former, we use
the classical definition~\citep{BuClSh97}: under the assumption that
$\V$ is equidimensional, this is the generic (and maximal) number of
intersection points of $\V$ with a linear space of the complementary
dimension. The notion of height is more technical: we give the
definition in Section~\ref{sec:height}.

Let then $\IP \subset \Q[\Y,\X]$ be the defining ideal of $\V$ and let
$\VS \subset \overline{\Q(\Y)}^n$ be the zero-set of $\IS=\IP \cdot
\Q(\Y)[\X]$. We make the following assumptions:
\begin{Ass} \label{ass2} {~}
  \begin{itemize}
  \item The algebraic set $\V$ is defined over $\Q$, equidimensional
    of dimension $m$ and the image of each irreducible component of
    $\V$ through $\Pi_0$ is dense in $\C^m$.
  \item The former point implies that $\VS$ has dimension 0; then,
    we assume that $\VS$ satisfies Assumption~\ref{ass1} over the base
    field $\Q(\Y)$.
  \end{itemize}
\end{Ass}
\noindent 
As a consequence, there exist polynomials $(T_1,\dots,T_n)$ in
$\Q(\Y)[\X]$ that generate the ideal~$\IS$; associated to them, we
also have the polynomials $(N_1,\dots,N_n)$ defined above, which are
in $\Q(\Y)[\X]$ as well. Then, Theorem~\ref{theo:main} below gives
degree and bit-size bounds for the polynomials $(T_1,\dots,T_n)$ and
$(N_1,\dots,N_n)$. As was said above, the degree bounds were already
in~\citep{DaSc04}; the bit-size aspects are new.

In the complexity estimates, we denote by $\V_\ell \subset
\C^{m+\ell}$ the Zariski-closure of the image of $\V$ through
$\Pi_\ell$, and let $d_{\V_\ell}$ and $h_{\V_\ell}$ be its degree and
height. The degree and height of $\V_\ell$ may be smaller than those
of $\V$, and cannot be larger (up to small parasite terms in the case
of height, see~\cite{KrPaSo01}). Next, for $\ell \le n$, we define
the projection
$$\begin{array}{cccc}
  \pi_\ell: & \overline{\Q(\Y)}^{n} & \to&  \overline{\Q(\Y)}^{\ell}\\
  & (x_1,\dots,x_n) & \mapsto & (x_1,\dots,x_\ell);
\end{array}$$
we let $\VS_\ell \subset \overline{\Q(\Y)}^{\ell}$ be 
the image of $\VS$ through $\pi_\ell$ and let $d_\ell \le d_{\V_\ell}$ be its degree.
Note that $\VS_\ell$ is obtained from $\V_\ell$ by the same process
that gives $\VS$ from $\V$. 

Finally, in the following theorem, the {\em height} $h(x)$ of a
non-zero integer $x$ denotes the real number $\log |x|$; it is a
measure of its bit-length. The height of a non-zero polynomial with
integer coefficients is the maximum of the heights of its non-zero
coefficients. Recall also that for polynomials in $\Z[\Y]$, gcd's and
lcm's are uniquely defined, up to sign.

\begin{Theo}\label{theo:main} 
  Suppose that $\V$ satisfies Assumption~\ref{ass2}.  For $1 \le \ell
  \le n$, let us write $N_\ell$ as
$$N_\ell = \sum_{\i} \frac{\gamma_{\i,\ell}}{\varphi_{\i,\ell}}X_1^{i_1} \cdots X_\ell^{i_\ell} + \frac {\gamma_\ell}{\varphi_\ell} X_\ell^{d_\ell}$$
and $T_\ell$ as
$$T_\ell = \sum_{\i} \frac{\beta_{\i,\ell}}{\alpha_{\i,\ell}}X_1^{i_1} \cdots X_\ell^{i_\ell} +  X_\ell^{d_\ell},$$
where:
\begin{itemize}
\item all multi-indices $\i=(i_1,\dots,i_\ell)$ satisfy $i_r <
  d_r$ for $r \le \ell$;
\item all polynomials $\gamma_{\i,\ell}$, $\varphi_{\i,\ell}$,
  $\gamma_{\ell}$ and $\varphi_\ell$, and $\beta_{\i,\ell}$,
  $\alpha_{\i,\ell}$, are in $\Z[\Y]$;
\item in $\Z[\Y]$, the equalities
  $\gcd(\gamma_{\i,\ell},\varphi_{\i,\ell}) =\gcd(\gamma_\ell,\varphi_\ell)=\gcd(\beta_{\i,\ell},\alpha_{\i,\ell})=\pm 1$
hold.
\end{itemize}
Then, all polynomials $\gamma_{\i,\ell}$ and $\gamma_\ell$,
$\varphi_{\i,\ell}$ and $\varphi_{\ell}$, as well as the lcm of all
$\varphi_{\i,\ell}$ and $\varphi_{\ell}$, have degree bounded by
$d_{\V_\ell}$ and height bounded by
 $${\cal H}_\ell \le 
2 h_{\V_\ell} +\big( (4 m +2)d_{\V_\ell} + 4m \big) \log(d_{\V_\ell}+1) + \big((10 m +  16)d_{\V_\ell} + 5\ell +2m\big)\log(m+\ell+3).$$
All polynomials $\beta_{\i,\ell}$ and $\alpha_{\i,\ell}$, as well as the lcm of all
$\alpha_{\i,\ell}$, have degree bounded by $2d_{\V_\ell}^2$ and height
bounded by
$$
\begin{array}{rcl}
{\cal H}'_\ell &\le  &
4d_{\V_\ell} h_{\V_\ell} + 3 d_{\V_\ell}^2  +  4\big((2m+1)d_{\V_\ell}^2+m(d_{\V_\ell}+1)\big) \log(d_{\V_\ell}+1)\\[1mm]
&&+\big((20 m+22)d_{\V_\ell}^2 + 5(d_{\V_\ell}+ \ell+m)\big) \log(m+\ell+3).
\end{array}
$$

\end{Theo}

\paragraph{Comments.} The first thing to note is that these bounds are
{\em polynomial} in the degree and height of $\V_\ell$, and are quite
similar to those obtained in~\citep{DaSc04} for the 0-dimensional case
(with $m=0$).  These results are actually simplified versions of more
precise estimates; they were obtained by performing (sometimes crude)
simplifications at various stages of the derivation. These
simplifications are nevertheless necessary to obtain compact formulas,
and the orders of magnitude of the results are unchanged: the bound
for $N_\ell$ is essentially of order $h_{\V_\ell}+d_{\V_\ell}$,
whereas that for $T_\ell$ has order
$(h_{\V_\ell}+d_{\V_\ell})d_{\V_\ell}$.

While we do not know about the sharpness of these results, they
reflect practical experience: in many cases, the polynomials
$(N_1,\dots,N_n)$ have much smaller coefficients than the polynomials
$(T_1,\dots,T_n)$; this was already pointed out for 0-dimensional
cases in~\citep{AlBeRoWo94,Rouillier99,DaSc04}.

These bounds are intrinsic, in that they do not depend on a given
system of generators of~$\IP$. As such, they behave well under
operations such as decomposition, due to the additivity of degree and
height of algebraic sets. Of course, if we are given bounds on
polynomials defining $\V$, it is possible to rewrite the previous
estimates in terms of these bounds, by means of the geometric and
arithmetic forms of B\'ezout's theorem. Suppose for instance that $\V$
is the zero-set of a system of $n$ polynomials of degree at most $d$,
with integer coefficients of height at most $h$; more generally, since
degree and height are additive, we could suppose that $\V$ consists of
one or several irreducible components of an algebraic set defined by
such a system.  The geometric B\'ezout inequality, and bounds on
degrees through projections~\citep{Heintz83} gives the inequality
$d_{\V_\ell} \le d^n$ for all $\ell$; similar results in an arithmetic
context~\citep{KrPaSo01} show that $h_{\V_\ell} \leq d^n(nh +
(4m+2n+3)\log(m+n+1))$ holds for all $\ell$. After substitution, this
gives
$${\cal H}_\ell = O\big ( d^n (nh+mn\log(d)+(m+n)\log(m+n))
\big)$$ 
and 
$${\cal H}'_\ell=O\big ( d^{2n} (n h +mn\log(d)+ (m+n)\log(m+n))  \big);$$
here, we write $f(m,n,d,h)=O(g(m,n,d,h))$ if there exists $\lambda >
0$ such that $f(m,n,d,h) \le \lambda g(m,n,d,h)$ holds for all $m,n,d,h$.
The main point is that the former grows roughly like $hd^{n}$, while
the latter grows like $hd^{2n}$.

To our knowledge, no previous result has been published on the
specific question of bounds in positive dimension. \citet{GaMi90} give
a derivation of degree bounds, which may be extended to give bit-size
estimates; these would however be of order $hd^{O(n^2)}$ at
best. Besides, such bounds would depend on a set of generators for the
ideal $\IP$ of $\V$.

As a consequence of our results, for many probabilistic arguments
involving say, computations modulo a prime $p$ (as is the case in
modular algorithms), choosing $p$ polynomial in the B\'ezout number is
enough to ensure a ``reasonable'' probability of success. We will
illustrate this in the last section of this paper.

\paragraph{Organization of the paper.}  
The paper is organized as follows. We start by recalling known
material on Chow forms (Section~\ref{sec:Chow}) and height theory
(Section~\ref{sec:height}). The next sections give a specialization
property for Chow forms, first in dimension~1
(Section~\ref{ssec:crv}), then more generally under
Assumption~\ref{ass2} (Section~\ref{sec:spec}). This will enable us to
predict suitable denominators for the polynomials $(N_1,\dots,N_n)$
and $(T_1,\dots,T_n)$, and give some first height estimates in
Section~\ref{sec:denom}; bounds on the numerators are obtained by
interpolation in Section~\ref{sec:main}, completing the
proof. Finally, Section~\ref{sec:example} illustrates the use of our
results by providing a probability analysis of a modular approach to
estimate the degrees in $(T_1,\dots,T_n)$.

\paragraph{Notation.}{~}
\begin{itemize}
\item If $F$ is a polynomial or a set of polynomials, $Z(F)$ denotes
  its set of zeros, in either an affine, a projective or a
  multi-projective space, this being clear from to the context.
\item Notation using superscripts such as $\U^i=U^i_0,\dots,U^i_{n}$
  {\em does not} denote powers.
\item As in the introduction, when speaking of an algebraic set
  defined over an unspecified field $k$, we will mainly use the
  notation $V$. For an algebraic set defined over $\Q$ and lying in
  some space such as $\C^{m+n}$, we will use the notation $\V$; the
  corresponding algebraic set defined over the rational function field
  $\Q(\Y)$ will be denoted $\VS$.
\end{itemize}


\section{Chow forms}\label{sec:Chow}

We review basic material on the Chow forms of an equidimensional
algebraic set. In this section, $k$ is a field of characteristic 0
and $V\subset {\overline k}^n$ is an equidimensional algebraic set
defined over $k$, of dimension~$r$.  Let $\X=X_1,\dots,X_n$ be the
coordinates in ${\overline k}^n$ and let $X_0$ be an homogenization
variable.  For $i=0,\dots,r$, let $\U^i=U^i_0,\dots,U^i_n$ be new
indeterminates, and associate them with the bilinear
forms $$L_i:~U^i_0 X_0+ \cdots + U^i_n X_n.$$ Let then $\overline V$
be the projective closure of $V$ in $\P^n(\overline k)$, and consider
the incidence variety $$W = \overline V \cap Z(L_0,\dots,L_{r})
\subset \overline V \times \underbrace{\P^{n}(\overline k) \times
  \cdots \times \P^{n}(\overline k)}_{r+1}.$$ The image of the
projection $W \to \P^{n}(\overline k) \times \cdots \times
\P^{n}(\overline k)$ is a hypersurface. A {\em Chow form} of $V$ is a
multi-homogeneous squarefree polynomial in $\overline
k[\U^0,\dots,\U^r]$ defining this hypersurface. All Chow forms thus
coincide up to a constant (non-zero) multiplicative factor in
$\overline k$; since $V$ is defined over $k$, Chow forms with
coefficients in $k$ exist.  The degree of a Chow form in the group of
variables $\U^i$ is the degree of $V$.

Note also the following fact: given an ideal $I$ of
$k[X_1,\dots,X_n]$, a field $k'$ containing $k$ and the extension $I'=I \cdot
k'[X_1,\dots,X_n]$, any Chow form of $V=Z(I) \subset {\overline k}^{\,n}$ is also
a Chow form of $V'=Z(I') \subset {\overline k'}^{\,n}$ (because the
image of the projection described above is defined over $k$).

Finally, consider the special case $r=0$, and let $I \subset
k[X_1,\dots,X_n]$ be the defining ideal of $V$. Then, the Chow forms
of $V$ are closely related to the characteristic polynomial of a
``generic linear form'' modulo $I$. To be more precise, let
$\U=U_0,\dots,U_{n}$ be the indeterminates of the Chow forms of $V$
(since the dimension $r$ equals $0$, we can drop the superscript
${}^0$ here). Over $\overline k$, the Chow forms of $V$ admit the
factorization
\begin{equation}\label{eq:facto}
c \prod_{x \in V} (U_0+ U_1 x_1 + \cdots + U_n x_n) \in \overline k[\U],
\end{equation}
where $c$ is in $\overline k$, and $x=(x_1,\dots,x_n)$. We will
distinguish two particular cases:
\begin{itemize}
\item taking $c=1$ in~\eqref{eq:facto}, we obtain what we will call
  the {\em monic} Chow form of $V$ (which has coefficients in $k$);
\item in the particular case $k=\Q(\Y)$, a {\em primitive} Chow form
  is a Chow form in $\Z[\Y,\U]=\Z[\Y][\U]\subset \Q(\Y)[\U]$, with
  content $\pm 1$ (the content is the gcd of the coefficients in
  $\Z[\Y]$). Primitive Chow forms are unique, up to sign.
\end{itemize}


\section{Absolute values and height}\label{sec:height}

Next, we recall the definitions and properties of absolute values and
heights for polynomials and algebraic sets. Our references
are~\citep{Lang83,McCarthy91,Philippon95,Sombra98,KrPaSo01}; our
presentation follows that of~\citet{DaSc04}, which itself is strongly
inspired by~\citet{KrPaSo01}. The proofs of all statements given here 
can be found in these references.


\subsection{Absolute values}\label{ssec:av}

An {\em absolute value} $v$ on a field $k$ is a multiplicative map $k
\to \R^+$, such that $v(a)= 0$ if and only if $a=0$, and for all $a,b
\in k^2$, we have $$v(a+b) \leq v(a)+v(b).$$ If the stronger
inequality $$v(a+b) \leq \max(v(a),v(b))$$ holds for all $a,b \in
k^2$, $v$ is called {\em non-Archimedean}, and {\em Archimedean}
otherwise. In any case, we will write $\ell_v(x) = \log(v(x))$, for $x
\ne 0$.

A family ${\sf M}_k$ of absolute values on $k$ verifies the {\it
  product formula} if for every $x \in k-\{0\} $, there are only a
finite number of $v$ in ${\sf M}_k$ such that $v(x) \neq 1$, and the
equality $$ \prod_{v \in {\sf M}_k} v(x)=1$$ holds. In this
case, we denote by ${\sf NA}_k$ and ${\sf A}_k$ the non-Archimedean
and Archimedean absolute values in ${\sf M}_k$, and write ${\sf
  M}_k=({\sf NA}_k,{\sf A}_k)$.  

Our first example of a valuated field is $k=\qu$. Let $\cal P$ be the
set of prime numbers, so that each $x$ in $\qu-\{0\}$ has the unique
factorization
$$x = \pm \prod_{p \in {\cal P}} p^{{\rm ord}_p(x)}.$$ For each prime
$p$, $x\mapsto v_p(x) = p^{-{\rm ord}_p(x)}$ defines a
non-Archime\-dean absolute value. Denoting $x \mapsto v_\infty(x)=|x|$
the usual Archimedean absolute value, we let ${\sf M}_\qu =( \{v_p,\,
p \in {\cal P}\},\ \{ v_\infty \})$, so that that ${\sf A}_\qu= \{
v_\infty\}$. One easily checks that ${\sf M}_\qu$ satisfies the
product formula.

The second example is $k=K(\Y)$, with $\Y=Y_1,\dots,Y_m$ and $K$ a
field. Let ${\cal S}$ be a set of irreducible polynomials in $K[\Y]$,
such that each $x$ in $K(\Y)-\{0\}$ has the factorization
$$x = c \prod_{S \in {\cal S}} S^{{\rm ord}_S(x)}, \ \ c \in
K.$$ Then each $S$ in ${\cal S}$ defines a non-Archimedean absolute
value $x \mapsto v_S(x)= e^{-\deg(S) \ {\rm ord}_S(x)}.$ An additional
non-Archimedean absolute value is given by $x \mapsto v_{\deg{}}(x) =
e^{\deg(x)},$ where $\deg(x)$ is defined as $\deg(n)-\deg(d)$, with
$n,d \in K[\Y]$ and $x=n/d$.  We define ${\sf M}_{K(\Y)} =( \{v_S,\, S
\in {\cal S}\} \cup \{v_{\deg{}} \},\ \emptyset)$, so that ${\sf
  A}_{K(\Y)}$ is empty. As before, ${\sf M}_{K(\Y)}$ satisfies the
product formula, though we will not use this fact here.

Finally, we can point out that the definition of height of an integer
we gave in the introduction fits with the definitions given here.
Indeed, in general, the height of a non-zero element $x$ in a field
$k$ with absolute value ${\sf M}_k$ that satisfy the product formula
is $h(x)=\sum_{v \in {\sf M}_k} \max (0,\ell_v(x))$; we recover
the particular case of the introduction for $k=\Q$. In particular, for
$x$ in $\Z-\{0\}$, $h(x)=\ell_{v_\infty}(x)$.


\subsection{Absolute values of polynomials}\label{ssec:avpol}

We next define absolute values and Mahler measures for polynomials
over the field $k$, and give a few useful inequalities. 

\paragraph{Absolute values.}
If $f$ is a non-zero polynomial with coefficients in $k$, for any
absolute value $v$ on $k$, we define the {\em $v$-adic absolute value}
of $f$ as
$$\ell_v(f)=  \max_{\beta} \{\ell_v(f_{\beta}) \},$$
where $f_\beta$ are the non-zero coefficients of $f$.
We give here a few obvious consequences of this definition, for
situations that will be considered later on. In the first example, $k$
is $\Q$, and we consider polynomials in $\Q[\Y]$.
\begin{itemize}
\item For $f$ in $\Q[\Y]$, $\ell_{v_p}(f) \le 0$ for all primes $p$ if
  and only if $f$ is in $\Z[\Y]$, and $\ell_{v_p}(f)=0$ for all primes
  $p$ if and only if $f$ is in $\Z[\Y]$ and has content $\pm 1$.
\item For $f$ in $\Z[\Y]$, $\ell_{v_\infty}(f)$ is the maximum of the
  heights of the non-zero coefficients of
  $f$.
\end{itemize}
In the next example, the base field $k$ is $\Q(\Y)$, and we consider
polynomials in $\Q(\Y)[\X]$ and the absolute values ${\sf M}_{K(\Y)}
=( \{v_S,\, S \in {\cal S}\} \cup \{v_{\deg{}} \},\ \emptyset)$
mentioned before.
\begin{itemize}
\item For $f$ in $\Q(\Y)[\X]$, $\ell_{v_S}(f) \le 0$ for all $S$ in
  ${\cal S}$ if and only if $f$ is in $\Q[\Y][\X]$.

\item For $f$ in $\Q[\Y][\X]$, $\ell_{v_{\deg{}}}(f)$ is the maximum
  of the degrees of the coefficients of $f$ (which are in $\Q[\Y]$).

\item By Gauss' Lemma, for $p$ prime, the $p$-adic absolute value
  $v_p$ defined on $\Q$ extend to a non-Archimedean absolute value
  $v_p$ on $\Q(\Y)$. For $f$ in $\Q[\Y][\X]$, $\ell_{v_p}(f) \le 0$
  holds for all primes $p$ if and only if $f$ is actually in
  $\Z[\Y][\X]$.
\end{itemize}

\paragraph{Mahler measures.} The following discussion is devoted to
the case $k=\Q$. In this case, we introduce Mahler measures, which are
closely related to Archimedean absolute values, but possess an extra
additivity property. If $f$ is in $\Q[\X^1,\dots,\X^r]$, where each
$\X^i$ is a group of $n$ variables, we define the {\em $r,n$-Mahler
  measure} ${\sf m}(f,r,n)$ as
$${\sf m}(f,r,n)=\int_{S_n^r} \log |f|\, \mu_n^r,$$ 
where $S_n \subset \C^n$ is the complex sphere of dimension $n$, and
$\mu _n$ is the Haar measure of mass 1 over $S_n$.

Remark that if $f$ depends on $r$ variables, the $r,1$-Mahler measure
${\sf m}(f,r,1)$ is the ``classical'' one, obtained by integration
over the product of $r$ unit circles.

\paragraph{Useful inequalities.} We conclude by giving basic
inequalities for absolute values and Mahler measures. If $v$ is
non-Archimedean over a field $k$, we have (Gauss' lemma)
\begin{itemize}
\item [$\mathbf{N_1}$] $\ell_v(f_1 f_2) = \ell_v(f_1) +\ell_v(f_2)$
  for any polynomials $f_1,f_2$ in $k[\Y]$.
\end{itemize}
If $k=\Q$ and $v=v_\infty$ is the Archimedean absolute value on $\Q$,
we have:
\begin{itemize}
\item [$\mathbf{A_1}$] $\ell_{v_\infty}(f) \le {\sf m}(f,r(n+1),1)+rd
  \log(n+2)$ if $f$ is a polynomial in $r$ groups of $n+1$ variables,
  of degree at most $d$ in each group.

\item [$\mathbf{A_2}$] ${\sf m}(f,r(n+1),1) \le {\sf m}(f,r,n+1)+rd \sum_{i=1}^n
  \frac 1{2i}$ if $f$ is a polynomial in $r$ groups of $n+1$
  variables, of degree at most $d$ in each group.

\item [$\mathbf{A_3}$] $\ell_{v_\infty}(f_1) + \ell_{v_\infty}(f_2)
  \le \ell_{v_\infty}(f_1f_2) + 4d\log(n+1)$, if $f_1$ and $f_2$ are
  polynomials in $n$ variables of degree at most $d$.
\end{itemize}


\subsection{Height of algebraic sets}\label{ssec:h}
 
We finally define heights of algebraic sets defined over $\Q$ (though
the construction can be extended to any field with a set of absolute
values satisfying the product formula). First, we note that as a
general rule, we will denote the degree of an algebraic set $\V$ by
$d_\V$, and its height by $h_\V$.

Let thus $\V \subset \C^{k}$ be an $m$-equidimensional algebraic set
defined over $\Q$ and let $\ChoV$ be a Chow form of $\V$ with
coefficients in $\Q$. We use the non-Archimedean absolute values and
Mahler measures of $\ChoV$ to define the height of $\V$. Let ${\sf
  M}_{\Q}=( \{v_p,\, p \in {\cal P}\},\ \{ v_\infty \})$ be the
absolute values on $\Q$ introduced before.  Then, as said above, we
let $d_{\V}$ be the degree of $\V$, and we define its height $h_{\V}$
as
$$h_{\V}=
\sum_{p \in {\cal P} } \ell_{v_p}(\ChoV) + {\sf m}(\ChoV,m+1,k+1)+
(m+1)d_{\V} \sum_{i=1}^{k} \frac{1}{2i}.$$ This is well-defined, as a
consequence of the product formula for ${\sf M}_{\Q}$. Then, the
definition extends by additivity to arbitrary algebraic sets.


\section{A specialization property}\label{ssec:crv}

Let $k$ be a field, and let $\varepsilon$ and $\X=X_1,\dots,X_n$ be
indeterminates over $k$. In this section, we work in the affine space
${\overline k}^{n+1}$, taking $\varepsilon$ and $\X$ for coordinates,
and we let $\pi$ be the projection
$$ \begin{array}{cccc} 
\pi: & {\overline k}^{n+1} &\to&  \overline k\\
&  (e,x_1,\dots,x_n) &\mapsto& e.
\end{array}$$
Let $V$ be an algebraic set in ${\overline k}^{n+1}$, defined over
$k$. We will show how to relate the Chow forms of the
``generic fiber'' of $\pi$ to those of the special fiber
above $e=0$. The results of this section will be used only in Section~\ref{sec:spec}.

We write $V$ as the union $V_0\, \cup\, V_1 \,\cup\, V_{\geq
  2}$, where:
\begin{itemize}
\item $V_0$ (resp. $V_1$) is the union of the irreducible components 
 of $V$ of dimension 0 (resp. of dimension 1);
\item $V_{\geq 2}$ is the union of the irreducible components of $V$
  of dimension at least 2;
\end{itemize}
remark that any of those can be empty. Let further $I \subset
k[\varepsilon,\X]$ be the ideal defining $V$, let $I^\star$ be
the extension of $I$ in $k(\varepsilon)[\X]$ and let $V^\star
\subset \overline{k(\varepsilon)}^n$ be the zero-set of
$I^\star$. Then, we introduce the following conditions:
\begin{description}
\item[$\bf G_1:$] The algebraic set $V^\star$ has dimension 0.
\item[$\bf G_2:$] The fiber $\pi^{-1}(0) \cap V$ has
  dimension 0.
\item[$\bf G_3:$] The fiber $\pi^{-1}(0) \cap V$ is
  contained in $V_1 \cup V_{\geq 2}$.
\end{description}
Let $\U=U_0,U_1,\dots,U_{n}$ be indeterminates, to be used for
Chow forms in dimension 0:
\begin{itemize}
\item Since $V^\star$ has dimension 0 by ${\bf G_1}$, its Chow forms
  are homogeneous polynomials in $\overline {k(\varepsilon)}[\U]$.
\item Let us denote by $W_0$ the fiber $\pi^{-1}(0) \cap
  V$ (the motivation for this notation appears below). Since $W_0$ has
  dimension 0 by ${\bf G_2}$, its Chow forms are homogeneous
  polynomials in $\overline k[\U]$.
\end{itemize}

\begin{Prop}\label{Prop:0}
  Suppose that ${\bf G_1}$, ${\bf G_2}$ and ${\bf G_3}$ hold. Let
  $C$ be a Chow form of $V^\star$, and
  suppose that $C$ belongs to the polynomial ring $k[\varepsilon,\U]
  \subset \overline{k(\varepsilon)}[\U]$. Then any Chow form of $W_0$ that
  belongs to $k[\U]$ divides $C(0,\U)$ in $k[\U]$.
\end{Prop}
\begin{proof} Let $W \subset V_1$ be the reunion of all
  $1$-dimensional components of $V$ whose image by $\pi$
  is dense in $\overline k$; we shall actually mainly be interested in
  $W$ in what follows. We start by the following easy lemma, which
  justifies our writing $W_0$ for the fiber $\pi^{-1}(0)
  \cap V$.

\begin{Lemma}
  The fiber $W_0=\pi^{-1}(0) \cap V$ is contained in $W$.
\end{Lemma}
\begin{proof}
  Let us write $W'$ for the reunion of all $1$-dimensional components
  of $V$ whose image by $\pi$ is not dense in $\overline
  k$; then $V_1$ is the union of $W$ and $W'$.  With this notation,
  Assumption ${\bf G_3}$ asserts that $W_0$ is contained in $W\, \cup
  \, W' \, \cup\, V_{\geq 2}$.
  
  The theorem on the dimension of fibers implies that all non-empty
  fibers of the restriction of $\pi$ to either $W'$ or $V_{\geq 2}$
  have positive dimension. So, the fact that $W_0$ has dimension 0
  (Assumption~${\bf G_2}$) implies that $W_0$ is contained in $W$.
\end{proof}

\noindent One easily checks that $W$ is defined over $k$; let then $J
\subset k[\varepsilon,\X]$ be its defining ideal, let $J^\star$
be the extension of $J$ in $k(\varepsilon)[\X]$ and let
$W^\star$ be the zero-set of $J^\star$. The following
lemma shows that the ``generic fibers'' of $\pi$
restricted to either $V$ or $W$ coincide.
\begin{Lemma}\label{Lemme:0}
  The equality $V^\star=W^\star$ holds.
\end{Lemma}
\begin{proof}
  We claim that all components of $V$ that are not in $W$ have a
  0-dimensional image through $\pi$:
  \begin{itemize}
  \item For the 1-dimensional components, this is true by definition of $W'$.
  \item Suppose that a component in $V_{\ge 2}$ has a dense image
    through $\pi$.  By the theorem on the dimensions of fibers, all
    fibers of $\pi$ on this component have positive dimension. These
    two points imply that the algebraic set $V^\star$ must have
    positive dimension as well. This contradicts Assumption~${\bf
      G_1}$.
  \end{itemize}
  Thus, we can write the equality $I=J \cap J'$, where $J'$ contains a
  non-zero polynomial in $k[\varepsilon]$.  Then, the extension of
  $J'$ to $k(\varepsilon)[\X]$ is the ideal $\langle 1\rangle$, so that
  $I^\star = J^\star$; this proves the statement.
\end{proof}

\noindent
By Lemma~\ref{Lemme:0}, the Chow forms of $V^\star$ and $W^\star$
coincide; they belong to $\overline {k(\varepsilon)}[\U]$. Let thus
$C$ be a Chow form of $W^\star$ that belongs to the
polynomial ring $k[\varepsilon,\U] \subset
\overline{k(\varepsilon)}[\U]$. We will now establish the proposition, 
that is, prove that any Chow form of
$W_0$ that belongs to $k[\U]$ divides $C(0,\U)$ in $k[\U]$.

The proof is inspired by that of~\citet[Prop.~1]{SaSo95}. We first
extend the coefficient field $k$, by adjoining to it the
indeterminates $U_1,\dots,U_n$; after this scalar extension, objects
that were previously defined over $k$ inherit the same denomination,
but using $\frak{fraktur}$ face: letting $\frak K$ be the rational
function field $k(U_1,\dots,U_n)$, we thus define the following
objects:
\begin{itemize}
\item $\frak J$ is the extension of $J$ in $\frak K[\varepsilon,\X]$
    and $\frak W$ is its zero-set.
   
    Still denoting by $\pi$ the projection on the first
    coordinate axis, we note that $\frak W$ inherits the geometric
    properties of $W$: it has pure dimension 1, and the restriction of
    $\pi$ to all its irreducible components is dominant.
    
  \item $\frak J^\star$ is the extension of $\frak J \subset
    \frak K[\varepsilon,\X]$ in $\frak K(\varepsilon)[\X]$. This is 
  a 0-dimensional ideal.
    
  \item $\frak W_0$ is the fiber $\pi^{-1}(0) \cap \frak
    W$. Since $W_0$ has dimension 0, $\frak W_0$ has dimension 0 as
    well.
\end{itemize}
 
\noindent The core of the proof is Lemma~\ref{Lemme:1.1} below. Recall
that $C \in k[\varepsilon,\U]$ is a Chow form of $W^\star$; we will
see $C$ in $ \frak K[\varepsilon,U_{0}]$, with $\frak
K=k(U_1,\dots,U_n)$. We also introduce the map
$$ \begin{array}{cccc}
\varphi:&  \frak W & \to & {\overline {\frak K}}^2 \\
&   (e,x_1,\dots,x_n) & \mapsto & (e,-U_1 x_1 -\cdots - U_n x_n).
\end{array}$$

 \begin{Lemma}\label{Lemme:1.1}
   Seen in $\frak K[\varepsilon,U_{0}]$, $C$ vanishes on the image of
   $\varphi$.
 \end{Lemma}
 \begin{proof}
   The closure of the image of $\varphi$ has dimension 1; we let $B$
   be a squarefree polynomial in $\frak K[\varepsilon,U_{0}]$ that
   defines this hypersurface. Note that $B$ does not admit any
   non-constant factor in $\frak K[\varepsilon]$, since all components
   of $\frak W$ have a dense image through $\pi$. Our goal
   is to show that $B$ divides $C$ in $\frak K[\varepsilon,U_{0}]$.
   
   Let us see $C$ in $\frak K[\varepsilon][U_0]$ and let $c \in \frak
   K[\varepsilon]$ be its leading coefficient. Since $C$ is a Chow
   form of $W^\star$, Proposition~4.2.7 in~\citep{CoLiOs98} shows
   that $C/c$ is the characteristic polynomial of the multiplication
   by $-U_1 X_1 -\cdots - U_n X_n$ modulo~$\frak J^\star$.

   On the other hand, Proposition 1 in~\citep{Schost03a} shows that
   $B/b$ is also the characteristic polynomial of the multiplication
   by $-U_1 X_1 -\cdots - U_n X_n$ modulo~$\frak J^\star$, where
   $b\in \frak K[\varepsilon]$ is the leading coefficient of $B$ seen
   in $\frak K[\varepsilon][U_{0}]$.  We deduce from these
   considerations the equality $Bc = Cb$ in $\frak
   K[\varepsilon,U_{0}]$; since $B$ admits no factor in $\frak
   K[\varepsilon]$, $b$ divides $c$ in $\frak K[\varepsilon]$, which
   proves our claim.
\end{proof}

\noindent Specializing $\varepsilon$ at $0$, we deduce that $C(0,\U)
\in \frak K[U_{0}]$ vanishes on the image of the map
$$ \begin{array}{cccc}
\varphi_0:&  \frak W_0 & \to & \overline {\frak K} \\
&  (x_1,\dots,x_n) & \mapsto & -U_1 x_1 -\cdots - U_n x_n.
\end{array}$$
Hence, it admits the polynomial $\prod_{x \in {\frak W_0}}(U_0+U_1 x_1 +\cdots
+ U_n x_n)$ as a factor. Note that this last
polynomial is the monic Chow form of $W_0$; note also that the division
takes place in $k[\U]$, since $C(0,\U)$ and this Chow form are in 
$k[\U]$, and the Chow form is monic in $U_0$.
Since all Chow forms 
of $W_0$ differ by a constant factor in $\overline k$, this concludes the proof of
Proposition~\ref{Prop:0}.
\end{proof}

\medskip\noindent Assumptions ${\bf G_1}$ and ${\bf G_2}$ will be easy
to ensure; to conclude, we give sufficient conditions that ensure
that ${\bf G_3}$ holds. 
\begin{Lemma} \label{lemma:8} Let $I' \subset k[\varepsilon,\X]$ be an ideal
  such that $V=Z(I')$ and suppose that there exist $F_1,\dots,F_n$ and
  $\Delta$ in $k[\varepsilon,\X]$ such that:
 \begin{itemize}
 \item the inclusions $\Delta I' \subset \langle F_1,\dots,F_n \rangle
   \subset I'$ hold;
 \item $\Delta(0,\X)$ is in $k-\{0\}$.
  \end{itemize}
  Then $V$ satisfies ${\bf G_3}$.
\end{Lemma}
\begin{proof}
  Let $V'$ be the Zariski closure of $V-Z(\Delta)$: each irreducible
  component of $V'$ is thus an irreducible component of $V$. Our
  assumptions imply that $V'$ coincides with the Zariski closure of
  $Z( F_1,\dots,F_n)-Z(\Delta).$ By Krull's theorem, all irreducible
  components of the zero-set $Z( F_1\dots,F_n)$ have dimension at
  least 1, so it is also the case for $V'$. To summarize, each
  irreducible component of $V'$ is a positive-dimensional irreducible
  component of $V$, so that $V'$ is contained in $V_1 \cup V_{\ge 2}$.

  Now, since $\Delta(0,\X)$ is in $k-\{0\}$, the fiber
  $\pi^{-1}(0)\cap V$ does not meet $Z(\Delta)$, so it
  is contained in~$V'$. This proves that $V$ satisfies Assumption
  ${\bf G_3}$.
\end{proof}


\section{Chow forms for the generic solutions}\label{sec:spec}

We consider now an $m$-equidimensional algebraic set $\V \subset
\C^{m+n}$ that satisfies Assumption~\ref{ass2}. As in the
introduction, we write the ambient coordinates as $\Y,\X$, with
$\Y=Y_1,\dots,Y_m$ and $\X=X_1,\dots,X_n$, and we recall that $\Pi_0$
is the projection $\C^{m+n} \to \C^{m}$.  We let $\IP$ be the ideal
defining $\V$, let $\IS$ be the extended ideal $\IP\cdot \Q(\Y)[\X]$
and let $\VS$ be the zero-set of $\IS$. In this section, we show how
to obtain a Chow form of $\VS$ starting from a Chow form of $\V$.

The Chow forms of $\V$ are polynomials in $(m+1)(m+n+1)$ variables,
which we write as $\U^i = U^i_0,\dots,U^i_{m+n}$, for
$i=0,\dots,m$. It will be helpful to have the following matrix
notation for these indeterminates:
$$
\U_{(0)} = \left [ \begin{array}{c} U^0_{0} \\ \vdots \\ U^m_{0}
  \end{array} \right ],
\U_{(\Y)} = \left [ \begin{array}{ccc} U^0_{1} & \dots & U^0_{m} \\
    \vdots &&\vdots\\
    U^m_{1} & \dots & U^m_{m}\end{array} \right ],\
\U_{(\X)} = \left [ \begin{array}{ccc} U^0_{m+1} & \dots & U^0_{m+n} \\
    \vdots &&\vdots\\
    U^m_{m+1} & \dots & U^m_{m+n}\end{array} \right ]. 
$$
This choice of variables corresponds to seeing these Chow forms as
polynomials defining the projection on $\P^{m+n}(\C) \times \cdots
\times \P^{m+n}(\C)$ of the incidence variety
$$\overline \V \cap Z(L_0,\dots,L_{m}) \subset \overline \V \times
\underbrace{ \P^{m+n}(\C) \times \cdots \times \P^{m+n}(\C)}_{m+1},$$
where $\overline \V$ is the projective closure of $\V$, where for
all $0 \leq i \leq m$, $L_i$ is the bilinear form
$$U^i_{0} T_{0} +  U^i_{1} Y_1 + \cdots + U^i_{m} Y_m+
U^i_{m+1} X_{1} + \cdots + U^i_{m+n} X_{n},$$ and where $T_0$ is an
homogenization variable. We will denote the Chow forms of $\V$ by
$\ChoV$.

Assumption~\ref{ass2} implies that $\VS \subset \overline{\Q(\Y)}^n$
has dimension 0, so we write $\U=U_0,\dots,U_{n}$ for the
indeterminates of the Chow forms of $\VS$. These Chow forms are in
$\overline{\Q(\Y)}[\U]$; however, we will be interested in those
belonging to the subring $\Z[\Y,\U]$ of $\overline{\Q(\Y)}[\U]$.

\citet{KrPaSo01} answer our question under an additional assumption.
Instead of requiring the restriction of $\Pi_0$ to $\V$ to be
dominant, their result requires the following stronger assumption:
\begin{Ass}\label{ass3}
  The restriction of $\Pi_0$ to $\V$ is finite, of degree the
  degree of $\V$.
\end{Ass}
\noindent Then, the following relation holds~\cite[Lemma 2.14]{KrPaSo01}.
\begin{Prop}\label{Prop:KrPaSo}
  Let $\ChoV \in \Z[\U^0,\dots,\U^m]$ be a Chow form of $\V$ and let
  $\Cho \in \Z[\Y,\U]$ be the polynomial obtained by performing the
  following substitution in $\ChoV$:
$$
\U_{(0)} \leftarrow \left [ \begin{array}{c} U_{0} \\ Y_1 \\ \vdots \\ Y_m
  \end{array} \right ],\
\U_{(\Y)} \leftarrow \left [
  \begin{array}{ccc} 0 & \dots & 0 \\ -1 &
    \dots & 0 \\
    \vdots & \ddots &\vdots\\
    0 & \dots & -1\end{array} \right ],\
\U_{(\X)} \leftarrow \left [ \begin{array}{ccc} U_1 & \dots & U_n \\ 0 & \dots & 0 \\
    \vdots &&\vdots\\
    0 & \dots & 0\end{array} \right ]. \ 
$$
If $\V$ satisfies Assumption~\ref{ass3}, then, seen in
$\overline{\Q(\Y)}[\U]$, $\Cho$ is a Chow form of $\VS$; in
particular, it is non-zero.
\end{Prop}

In our more general setting, one can still perform this substitution,
but the result might be zero. For instance, the algebraic set $\V$
defined by the system in $\Q[Y_1,Y_2,X_1,X_2]$
$$X_1+1+Y_1 X_2=0, \quad X_2+Y_2 X_1=0$$
satisfies Assumption~\ref{ass2} but not Assumption~\ref{ass3}. Indeed, 
 since
$$\langle X_1 + 1 + Y_1 X_2 , X_2 + Y_2 X_1 \rangle \cap \Q[Y_1,Y_2]= \langle 0 \rangle,
$$
the projection of $\V$ on the $(Y_1,Y_2)$-space is dense, and the
associated triangular set in $\Q(Y_1,Y_2)[X_1,X_2]$ is \sloppy
$T_1(X_1)=X_1 + 1/(1-Y_1 Y_2)$ and $T_2(X_1,X_2)=X_2 + Y_2 X_1$; this
gives Assumption~\ref{ass2}. To see why Assumption~~\ref{ass3} is not
verified by this example, note that for any $(y_1,y_2) \in \C^2$ with
$y_1y_2\ne 1$, the fiber $\Pi_0^{-1}(y_1,y_2)$ has cardinality 1
(whereas $\V$ has degree 4); if $y_1y_2=1$, the fiber is empty. As it
turns out, the Chow forms of $\V$ are polynomials in 15 variables,
having 6648 monomials, and performing the substitution of
Proposition~\ref{Prop:KrPaSo} in them gives zero.

The following theorem shows how to bypass this difficulty, by
providing a suitable multiple of a Chow form of $\VS$. To this effect,
we need to introduce a new indeterminate $\varepsilon$.

\begin{Theo}\label{Prop:1}
  Let $\ChoV \in \Z[\U^0,\dots,\U^m]$ be a Chow form of $\V$ and let
  $\ChoV_\varepsilon \in \Z[\Y,\U,\U^1,\dots,\U^m,\varepsilon]$ be the
  polynomial obtained by performing the following substitution
  in~$\ChoV$:
 $$
\U_{(0)} \leftarrow \left [ \begin{array}{c} U_{0} \\ Y_1 \\ \vdots \\ Y_m \end{array} \right ],\ \
\U_{(\Y)} \leftarrow \left [
  \begin{array}{ccc} 0 & \dots & 0 \\ -1 &
    \dots & 0 \\
    \vdots & \ddots &\vdots\\
    0 & \dots & -1\end{array} \right ],\ \
\U_{(\X)} \leftarrow \left [ \begin{array}{ccc} U_1 & \dots &
    U_n
    \\ \varepsilon U^1_{m+1}  & \dots & \varepsilon U^1_{m+n} \\
    \vdots &&\vdots\\
      \varepsilon U^m_{m+1} & \dots &  \varepsilon U^m_{m+n}\end{array} \right
    ].
$$
Then, $\ChoV_\varepsilon$ is not zero. Let $\ChoV_0 \in
\Z[\Y,\U,\U^1\dots,\U^m]$ be the coefficient of lowest degree in
$\varepsilon$ of $\ChoV_\varepsilon$, and let finally $\Cho \in
\Z[\Y,\U]$ be a primitive Chow form of $\VS$. Then $\Cho$ divides
$\ChoV_0$ in $\Z[\Y,\U,\U^1,\dots,\U^m]$.
\end{Theo}

\paragraph{Ingredients used in the proof.}
The proof will occupy the remainder of this section. Let us start by
explaining the ingredients of it. We will apply a generic change of
variables, to get back under Assumption~\ref{ass3}; introducing the
matrix of this change of variables will require to work over a purely
transcendental extension of $\Q$.
\begin{itemize}
\item In the first step of the proof, we will work over the field
  $\L=\Q(\T^1,\dots,\T^m,\varepsilon)$, where $\T^i=T^i_1,\dots,T^i_n$
  are new indeterminates; we will use $\T^1,\dots,\T^m$ and
  $\varepsilon$ to perform our change of variables.
\item In the last step of the proof, we let $\varepsilon \to 0$, by
  working over the coefficient fields $\K=\Q(\T^1,\dots,\T^m,\Y)$ and
  $\M=\Q(\T^1,\dots,\T^m,\varepsilon,\Y)$, so that $\M=\K(\varepsilon)
  = \L(\Y)$. The connection will be done using the results of
  Section~\ref{ssec:crv}.
\end{itemize}
This lattice of fields is represented in the following diagram:
$$
\xymatrix{
                & \L=\Q(\T^1,\dots,\T^m,\varepsilon)\ar[dr]    \\
\Q\ar[ur]\ar[dr] &                        & \M=\Q(\T^1,\dots,\T^m,\varepsilon,\Y)=\K(\varepsilon)=\L(\Y). \\
                & \K=\Q(\T^1,\dots,\T^m,\Y)\ar[ur] \\
}
$$


\subsection{Application of a generic change of variables}\label{ssec:first}

First, we work over $\L=\Q(\T^1,\dots,\T^m,\varepsilon)$.  To recover
Assumption~\ref{ass3}, we define the following new coordinates for
${\overline \L}^{m+n}$:
\begin{equation}\label{eq:prime}
\left [ \begin{array}{c} \widetilde X_1 \\ \vdots \\ \widetilde X_n \end{array} \right ] = 
\left [ \begin{array}{c} X_1 \\ \vdots \\ X_n
\end{array} \right ] \quad  {\rm~and~} \quad
\left [ \begin{array}{c} \widetilde Y_1 \\ \vdots \\ \widetilde Y_m 
\end{array} \right ] = 
\left [ \begin{array}{c} Y_1 \\ \vdots \\ Y_m 
\end{array} \right ] + 
\left [ \begin{array}{ccc}  \varepsilon T^1_1 & \dots &  \varepsilon T^1_n \\
\vdots &&\vdots\\
\varepsilon T^m_1 & \dots &  \varepsilon T^m_n\end{array} \right ]
\left [ \begin{array}{c} X_1 \\ \vdots \\ X_n
\end{array} \right ].\end{equation}
In all that follows, we write for short $\widetilde \Y=\widetilde Y_1,\dots,\widetilde Y_m$ and
$\widetilde \X=\widetilde X_1,\dots,\widetilde X_n$. 
Then, we define the ideal $\J$ as 
$$\J= \langle \ F(\widetilde \Y,\widetilde \X) \ | \ F \in \IP\ \rangle \ \ \subset \ \ 
\L[\Y,\X],$$ and we let $\W \subset {\overline \L}^{m+n}$ be the
zero-set of $\J$. Note that $\W$ is equidimensional of dimension
$m$, and has the same degree as $\V$.

Since $\W$ is in generic coordinates, we will apply
Proposition~\ref{Prop:KrPaSo} to obtain a Chow form of its ``generic
solutions''. Recall the definition $\M=\L(\Y)$; we let $\JS$ be the
extension of $\J$ in the polynomial ring $\L(\Y)[\X]=\M[\X]$, and
denote by $\WS$ its set of solutions. Then, the first step of the
proof of Theorem~\ref{Prop:1} is the following.
\begin{Prop}\label{Lemme:4}
  The algebraic set $\WS$ has dimension 0. Let further
  $\ChoV\in \Z[\U^0,\dots,\U^m]$ be a Chow form of $\V$, and let
  $\ChoV^\star$ be the polynomial in
  $\Z[\Y,\U,\T^1,\dots,\T^m,\varepsilon]$ obtained by performing the
  following substitution in $\ChoV$:
$$
\U_{(0)} \leftarrow \left [ \begin{array}{c} U_{0} \\ Y_1 \\ \vdots \\ Y_m
  \end{array} \right ],\ \
\U_{(\Y)} \leftarrow \left [
  \begin{array}{ccc} 0 & \dots & 0 \\ -1 &
    \dots & 0 \\
    \vdots & \ddots &\vdots\\
    0 & \dots & -1\end{array} \right ],\ \ 
\U_{(\X)} \leftarrow \left [ \begin{array}{ccc} U_1 & \dots &
    U_n
    \\ \varepsilon T_1^1  & \dots & \varepsilon T_n^1 \\
    \vdots &&\vdots\\
      \varepsilon T_1^m & \dots &  \varepsilon T_n^m\end{array} \right
    ].
$$
Then, seen in $\M[\U]$, $\ChoV^\star$ is a Chow form of $\WS$;
in particular, it is non-zero.
\end{Prop}
This subsection is devoted to give a proof of this proposition.  The
key element is the following lemma.
\begin{Lemma}
  The algebraic set $\W$ satisfies Assumption~\ref{ass3}; in
  particular, $\WS$ has dimension~0.
\end{Lemma}
\begin{proof}
  Let $d_\V$ be the degree of $\V$.  By definition of the degree,
  there exists a Zariski-dense subset $\Gamma$ of $\C^{m(m+n+1)}$ such
  that for all choices of $(u^i_0,\dots,u^i_{m+n})_{1\le i \le m}$ in
  $\Gamma$, the algebraic set
  \begin{equation}\label{eq:2}
    \V  \cap  Z(\{u^i_{0} + u^i_{1} Y_1 + \cdots + u^i_{m} Y_m 
    + u^i_{m+1} X_1 + \cdots + u^i_{m+n} X_n   
    \}_{1 \leq i \leq m})
    \end{equation}
    has dimension 0 and cardinality $d_\V$, and furthermore the determinant
  $$\left | \begin{array} {ccc}
     u^1_{1} & \dots & u^1_{m} \\
       \vdots &  & \vdots \\
     u^m_{1} & \dots & u^m_{m} 
   \end{array} \right |$$
 is non-zero. Thus, there exists 
 \begin{equation}
   \label{1}
\left [\begin{array}{c}   
u^1_{0} \\
\vdots \\
u^m_{0} 
   \end{array} \right ]
\quad \text{and}\quad
\left [ \begin{array} {ccc}
     u^1_{1} & \dots & u^1_{m} \\
       \vdots &  & \vdots \\
     u^m_{1} & \dots & u^m_{m} 
   \end{array} \right ]
 \end{equation}
 in $\Q^{m(m+1)}$ and an open dense subset $\Gamma'$ of $\C^{mn}$ such
 that for all
$$\left [\begin{array}{cccc}   
     u^1_{m+1} & \dots & u^1_{m+n} \\
       \vdots &  & \vdots  \\
     u^m_{m+1} & \dots & u^m_{m+n} 
   \end{array} \right ]$$
 in $\Gamma'$, the former property holds.
 We keep the quantities of~\eqref{1} fixed, and we define
 $y_1,\dots,y_m$ by
$$\left [\begin{array}{c}   
y_1\\
\vdots \\
y_m
   \end{array} \right ]
= -
\left [\begin{array}{ccc}   
     u^1_{1} & \dots & u^1_{m} \\
       \vdots &  & \vdots \\
     u^m_{1} & \dots & u^m_{m} 
   \end{array} \right ]^{-1} 
\left [\begin{array}{c}   
u^1_{0} \\
\vdots \\
u^m_{0} 
   \end{array} \right ].
$$
Besides, we let $\Lambda \subset \C^{mn}$ be the image of
$\Gamma'$ through the map
$$
\left [\begin{array}{ccc}   
     u^1_{m+1} & \dots & u^1_{m+n} \\
       \vdots &  & \vdots   \\
     u^m_{m+1} & \dots & u^m_{m+n}
   \end{array} \right ]
\mapsto  -
\left [\begin{array}{ccc}   
     u^1_{1} & \dots & u^1_{m} \\
       \vdots &  & \vdots \\
     u^m_{1} & \dots & u^m_{m} 
   \end{array} \right ]^{-1} 
\left [\begin{array}{cccc}   
     u^1_{m+1} & \dots & u^1_{m+n} \\
       \vdots &  & \vdots  \\
     u^m_{m+1} & \dots & u^m_{m+n} 
   \end{array} \right ];$$
 so that $\Lambda$ is dense in $\C^{mn}$. For any choice of
 $(\t^i=(t^i_1,\dots,t^i_n))_{1 \le i \le m}$ in $\Lambda$, the algebraic set
  \begin{equation*}
    \V \cap Z(\{Y_i - t^i_1 X_1 - \cdots - t^i_n X_n - y_i\}_{1 \leq i
    \leq m}) \ \subset \ \C^{m+n}
 \end{equation*}
 has dimension 0 and cardinality $d_\V$. Let finally $\Lambda'
 \subset \C^{mn+1}$ be the preimage of $\Lambda$ by the surjective map
 $(\t^1,\dots,\t^m,e) \mapsto (e\t^1,\dots,e\t^m)$, where
 $e\t^i=(et^i_1,\dots,et^i_n)$. Then, $\Lambda'$ is dense in
 $\C^{mn+1}$ and for all $(\t^1,\dots,\t^m,e)$ in $\Lambda'$, the
 algebraic set
\begin{eqnarray*}
   && \V\cap Z(\{Y_i - e t^i_1 X_1 - \cdots - e t^i_n X_n - y_i\}_{1 \leq i
    \leq m}) \ \subset \ \C^{m+n}
\end{eqnarray*}
has dimension 0 and cardinality $d_\V$. Since this property holds
for $(\t^1,\dots,\t^m,e)$ in a dense subset of $\C^{mn+1}$, we deduce
from~\cite[Prop. 1]{Heintz83} that the algebraic set defined over
$\L=\Q(\T^1,\dots,\T^m,\varepsilon)$ by
\begin{equation*}
  \V\cap Z(\{Y_i - \varepsilon T^i_1 X_1 - \cdots - \varepsilon T^i_n X_n - y_i\}_{1 \leq i
    \leq m})\ \subset\ \overline \L^{m+n}
\end{equation*}
has dimension 0 and cardinality $d_\V$.  But this algebraic set is
isomorphic through the change of variables $\Y \leftrightarrow
\widetilde \Y$ to 
\begin{equation*}
  \W \cap Z(\{Y_i - y_i\}_{1 \leq i\leq m}) \ \subset\ \overline \L^{m+n},
\end{equation*}
which is the fiber $\Pi_0^{-1}(y_1,\dots,y_m)\cap \W$.  

To summarize, $\W$ is an $m$-equidimensional algebraic set, and the
fiber $\Pi_0^{-1}(y_1,\dots,y_m)\cap \W$ has a cardinality equal to
the degree of $\W$. The first point of~\cite[Lemma 2.14]{KrPaSo01}
implies that under these conditions, $\W$ satisfies
Assumption~\ref{ass3}.
\end{proof}

We can now conclude the proof of Proposition~\ref{Lemme:4}.  If
$\ChoV$ is a Chow form of $\V=Z(\IP)$, since
$\L=\Q(\T^1,\dots,\T^m,\varepsilon)$, $\ChoV$ is also a Chow form of
the algebraic set defined by the extension of $\IP$ in $\L[\Y,\X]$ (we
mentioned this fact in Section~\ref{sec:Chow}). Since $\W$ is obtained
by applying a linear change of variables to this algebraic set, we can
deduce a Chow form of $\W$ by changing the variables in $\ChoV$: Let
$\widetilde{\U}_{(\X)}$ be the matrix
$$
\left [ \begin{array}{ccc} U^0_{m+1} & \dots & U^0_{m+n} \\
    \vdots &&\vdots\\
    U^m_{m+1} & \dots & U^m_{m+n}\end{array} \right ]\
-  \left [ \begin{array}{ccc} U^0_{1} & \dots & U^0_{m} \\
    \vdots &&\vdots\\
    U^m_{1} & \dots & U^m_{m}\end{array} \right ]
\left [ \begin{array}{ccc}  \varepsilon T^1_1 & \dots & \varepsilon T^1_n \\
    \vdots &&\vdots\\
    \varepsilon T^m_1 & \dots & \varepsilon T^m_n\end{array} \right ];
$$
then $\ChoV(\U_{(0)},\U_{(\Y)},\widetilde \U_{(\X)})$ is a Chow form of
$\W$. Now, we apply Proposition~\ref{Prop:KrPaSo} to $\W$,
which is legitimate by the previous lemma; this gives the announced
result.


\subsection{Setup for the specialization $\varepsilon=0$}\label{ssec:second}

The final part of the proof consists in letting $\varepsilon=0$ in the
previous result; this will be done in the next subsection, by applying
the results of Section~\ref{ssec:crv}. The purpose of this subsection
is to prove that the necessary assumptions hold.  We work here using
$\K=\Q(\T^1,\dots,\T^m,\Y)$ as our base field.  Using the notation of
Equations~(\ref{eq:prime}), we define the ideal $\KP$ as
$$\KP = \langle \ F(\widetilde \Y,\widetilde \X) \ | \ F \in \IP\ \rangle \ \subset
\K[\varepsilon,\X].$$ Let $\WP \subset \overline{\K}^{n+1}$ be the
zero-set of $\KP$. As in Section~\ref{ssec:crv}, we write $\pi$ for
the projection map $(e,x_1,\dots,x_n) \mapsto e$; our purpose is to
establish the following proposition.

\begin{Prop}\label{Prop:g1g2g3}
  The algebraic set $\WP$ satisfies Assumptions ${\bf G_1}$, ${\bf
    G_2}$ and ${\bf G_3}$ of Section~\ref{ssec:crv}.
\end{Prop}
Remark that there exist polynomials $F_1,\dots,F_n$ in $\Q[\Y,\X]$
that generate the extended ideal $\IP\cdot \Q(\Y)[\X]$, since this
ideal is 0-dimensional (actually, we can take the polynomials
$T_1,\dots,T_n$, whose existence is guaranteed by
Assumption~\ref{ass2}, and clear their denominators). We will first
relate the ideals $\KP$ and $\langle F_1(\widetilde \Y,\widetilde
\X),\dots,F_n(\widetilde \Y,\widetilde \X) \rangle$ in
$\K[\varepsilon,\X]$.

\begin{Lemma} \label{lemma:8b}
There exists $\Delta \in \K[\varepsilon,\X]$ such that:
 \begin{itemize}
 \item the inclusions $\Delta \KP \subset \langle
   F_1(\widetilde \Y,\widetilde \X),\dots,F_n(\widetilde \Y,\widetilde \X) \rangle \subset \KP$ hold;
 \item $\Delta(0,\X)$ is in $\Q[\Y] \subset \K$ and is non-zero.
  \end{itemize}
\end{Lemma}
\begin{proof}
  Let $f_1,\dots,f_s \in \Q[\Y,\X]$ be generators of $\IP$. By
  construction, all polynomials $F_j$, for $j=1,\dots,n$, can be
  expressed through equalities of the form
  $$F_j = \sum_{i=1}^s h_{i,j} f_i,$$
  for some $h_{i,j}$ in $\Q(\Y)[\X]$. Clearing denominators, these
  equalities can be rewritten as
  $$\gamma_j F_j = \sum_{i=1}^s H_{i,j} f_i,$$
  for some coefficients $H_{i,j}$ in $\Q[\Y,\X]$ and $\gamma_j$ in
  $\Q[\Y]$. Assumption~\ref{ass2} on $\V$ then implies that $F_j$
  itself belongs to the ideal $\IP$; the rightmost inclusion of the
  first point follows, after applying the change of variable
  in~\eqref{eq:prime}.

  Conversely, each polynomial $f_i$ belongs to the ideal $\IP\cdot
  \Q(\Y)[\X]$, so that for $i=1,\dots,s$, there is an equality of the
  form $$f_i = \sum_{j=1}^n a_{i,j} F_j,$$ for some $a_{i,j}$ in
  $\Q(\Y)[\X]$.  Clearing denominators, we can rewrite this equality as
  $$\delta_i f_i = \sum_{j=1}^n A_{i,j} F_j,$$
  for some $A_{i,j}$ in $\Q[\Y,\X]$ and $\delta_i$ non-zero in
  $\Q[\Y]$. Taking the least common multiple of all $\delta_i$, we
  finally obtain expressions of the form
  $$\delta f_i = \sum_{j=1}^n B_{i,j} F_j,$$ for some 
  $B_{i,j}$ in $\Q[\Y,\X]$ and $\delta$ in $\Q[\Y]$.  Define $\Delta =
  \delta(\widetilde \Y) \in \K[\varepsilon,\X]$, and note that
  $\Delta(0,\X) = \delta \in \Q[\Y]$. Then, we deduce the equalities
  $$\Delta\, f_i(\widetilde \Y,\widetilde \X) = \sum_{j=1}^n B_{i,j}(\widetilde \Y,\widetilde \X) F_j(\widetilde \Y,\widetilde \X),$$
  so that $$\Delta\, f_i(\widetilde \Y,\widetilde \X) \in \langle
  F_1(\widetilde \Y,\widetilde \X),\dots,F_n(\widetilde \Y,\widetilde \X) \rangle
  $$
  for all $i$; this finishes the proof.
\end{proof}
  
\medskip\noindent We can then conclude the proof of Proposition~\ref{Prop:g1g2g3}.
\begin{itemize}
\item The extension of $\KP \subset \K[\varepsilon,\X]$ in
  $\K(\varepsilon)[\X]=\M[\X]$ is the ideal $\langle F(\widetilde
  \Y,\widetilde \X) \ | \ F \in \IP \rangle$ of $\M[\X]$; it is thus
  the ideal $\JS$ defined in the previous subsection. This ideal has
  dimension 0, so that $\WP$ satisfies ${\bf G_1}$.
\item The fiber $\pi^{-1}(0) \cap \WP$ is obtained by adding
  $\varepsilon=0$ to the defining equations of $\WP$; it is thus
  defined by the ideal $\IP \cdot \K[\X]$. Since $\K$ is built by
  adjoining new transcendentals to $\Q(\Y)$, and since $\IP \cdot
  \Q(\Y)[\X]$ has dimension 0, $\pi^{-1}(0) \cap \WP$ has dimension
  0. Thus, $\WP$ satisfies ${\bf G_2}$.
\item Lemmas~\ref{lemma:8} and~\ref{lemma:8b} establish that $\WP$ satisfies ${\bf
    G_3}$.
\end{itemize}


\subsection{Conclusion}

We will now conclude the proof of Theorem~\ref{Prop:1}. Let
$\ChoV\in\Z[\U^0,\dots,\U^m]$ be a Chow form of $\V$, and let
$\ChoV_\varepsilon \in \Z[\Y,\U,\T^1,\dots,\T^m,\varepsilon]$ be the
polynomial obtained by performing the following substitution in
$\ChoV$:
$$
\U_{(0)} \leftarrow \left [ \begin{array}{c} U_{0} \\ Y_1 \\ \vdots \\ Y_m
  \end{array} \right ],\ \ 
\U_{(\Y)} \leftarrow \left [
  \begin{array}{ccc} 0 & \dots & 0 \\ -1 &
    \dots & 0 \\
    \vdots & \ddots &\vdots\\
    0 & \dots & -1\end{array} \right ],\ \ 
\U_{(\X)} \leftarrow \left [ \begin{array}{ccc} U_1 & \dots &
    U_n
    \\ \varepsilon T_1^1  & \dots & \varepsilon T_n^1 \\
    \vdots &&\vdots\\
      \varepsilon T_1^m & \dots &  \varepsilon T_n^m\end{array} \right
    ].
$$
Then, by Proposition~\ref{Lemme:4}, seen in $\M[\U]$,
$\ChoV_\varepsilon$ is a Chow form of $\WS$ (and so, is
non-zero). Besides, if $d$ is the valuation of $\ChoV_\varepsilon$ in
$\varepsilon$, then
$\ChoV'_\varepsilon=\ChoV_\varepsilon/\varepsilon^d$ is also a Chow
form of $\WS$, since $\varepsilon$ belongs to the base field $\M$.

Let now $\ChoV_0 \in \Z[\Y,\U,\T^1,\dots,\T^m]$ be the coefficient of
lowest degree in $\varepsilon$ of $\ChoV_\varepsilon$; it is thus
obtained by letting $\varepsilon=0$ in $\ChoV'_\varepsilon$.  Recall
that the extension of $\KP$ to $\M[\X]$ is the defining ideal $\JS$ of
$\WS$. Besides, by Proposition~\ref{Prop:g1g2g3}, $\WP=Z(\KP)$
satisfies Assumptions ${\bf G_1}$, ${\bf G_2}$ and ${\bf G_3}$ of
Proposition~\ref{Prop:0}. We deduce from that proposition that any
Chow form of the fiber $\pi^{-1}(0)\cap \WP$ divides $\ChoV_0$ in
$\K[\U]$.

As mentioned in the proof of Proposition~\ref{Prop:g1g2g3}, the fiber
$\pi^{-1}(0)\cap \WP$ is defined by the extension
of $\IP\cdot \Q(\Y)[\X]$ in $\K[\X]$.
Let thus $\Cho \in \Q(\Y)[\U]$ be a Chow form of $\IP\cdot
\Q(\Y)[\X]$. By the former remark, $\Cho$ is a Chow form of $\IP\cdot
\K[\X]$, so it divides $\ChoV_0$ in
$\K[\U]=\Q(\Y,\T^1,\dots,\T^m)[\U]$.
If we additionally impose that $\Cho$ is a {\it primitive} Chow form,
so that in particular it belongs to $\Z[\Y,\U]$, then one deduces that
$\Cho$ divides $\ChoV_0$ in $\Z[\Y,\T^1,\dots,\T^m,\U]$. This finishes
the proof, up to formally replacing the indeterminates $\T^i$ by the
indeterminates $\U^i$ appearing in the statement of
Theorem~\ref{Prop:1}.


\section{Predicting a denominator}\label{sec:denom}

We continue with the notation of the previous section, and study the
polynomials $(N_1,\dots,N_n)$ and $(T_1,\dots,T_n)$ of $\Q(\Y)[\X]$,
that were defined in the introduction. We reuse some notation from the
introduction, such as the degree $d_\V$ and the height $h_\V$ of $\V$,
and the degrees $(d_1,\dots,d_n)$ of the polynomials
$(T_1,\dots,T_n)$.  The notation of Section~\ref{sec:height} is in use
as well. We will use the constant
\begin{eqnarray*}
\G_{n} &=&  1 + 2{\sum}_{i \leq n-1} (d_i-1) \\
\end{eqnarray*}
Because $d_1 \cdots d_n \le d_\V$, one easily deduces the upper bound
$\G_{n} \le  2 d_\V$.

A first goal in this section is to predict suitable ``common
denominators'' for the polynomials $(N_1,\dots,N_n)$. We also wish to
do the same for the polynomials $(T_1,\dots,T_n)$, but this is not as
straightforward; for this reason, we are going to introduce a slightly
modified version of $(T_1,\dots,T_n)$, which will be more handy.  For
$i=1,\dots,n$, let us define the iterated resultant
$$e_i = \res(\cdots \res( \frac{\partial T_i}{\partial X_i}, T_i, X_i),\cdots,T_1, X_1) \in \Q(\Y),$$
so that for instance $e_1$ is the discriminant of $T_1$. We define the
polynomials $\num_1,\dots,\num_n$ by $\num_\ell = e_1 \cdots
e_{\ell-1} T_\ell$ for $\ell \le n$. As it turns, these polynomials
are easier to handle than the polynomials $T_\ell$, and the bit-length
information we wish to obtain for $T_\ell$ can easily be recovered
from $\num_\ell$.

The Chow forms of $\VS$ are polynomials in $\Q(\Y)[\U]$, where
$\U=U_0,\dots,U_n$ are new indeterminates. We will especially be
interested in a {\em primitive} Chow form of $\VS$; recall that it is
unique, up to sign. Informally, the denominator we seek will be the
leading coefficient of one of these primitive Chow forms. Formally,
choosing one the two possible signs, we let $\Cho \in \Z[\Y,\U]$ be a
primitive Chow form of $\VS$ and we let $a_n \in \Z[\Y]$ be the
coefficient of $U_0^{d_n}$ in $\Cho$.

\begin{Prop}\label{Prop:DH}
The following holds:
\begin{itemize}
\item $a_n \ne 0$;
\item $\ell_{v_\infty}(a_n) \le h_{\V} +  5(m+1)d_{\V}\log(m+n+2)$;
\item $\deg(a_n) \le d_{\V}$;
\item $a_n N_n$ is in $\Z[\Y,\X]$, with $\deg(a_n N_n,\Y) \le d_{\V}$;
\item $a_n^{\G_n} \num_n$ is in $\Z[\Y,\X]$, with $\deg(a_n^{\G_n} \num_n,\Y) \le \G_n d_{\V}$.
\end{itemize}
\end{Prop}
\noindent The first point is obvious: since $\VS$ has dimension 0,
Equation~\eqref{eq:facto} shows that the coefficient of $U_0^{d_n}$ in
$\Cho$ is non-zero. Then, Subsection~\ref{ssec:dh} will prove the degree
and height estimates for $a_n$; Subsection~\ref{ssec:val} will
prove the last assertions by means of valuation estimates. 

Finally, remark that in Proposition~\ref{Prop:DH}, we deal only with $N_n$ and
$\num_n$. However, this result implies analogue results for all $N_\ell$ and
$\num_\ell$, by replacing $\V$ by $\V_\ell$ and $\VS$ by $\VS_\ell$.


\subsection{Degree and height bounds for the primitive Chow form}\label{ssec:dh}

To prove the second and third points of Proposition~\ref{Prop:DH}, we
actually prove a similar estimate for the whole primitive Chow
form $\Cho$ of $\VS$.
\begin{Prop}\label{Prop:DH1}
  The primitive Chow form $\Cho$ of $\VS$ satisfies $\deg(\Cho,\Y) \le
  d_{\V}$ and $\ell_{v_\infty}(\Cho) \le h_{\V} + 5(m+1)d_{\V}\log(m+n+2).$
\end{Prop}
First, we recall from~\cite[Lemma~3]{Schost03a} that a primitive Chow
form of $\VS$ has degree in $\Y$ at most $d_{\V}$: this handles the
claimed degree bound. To deal with the height aspect, we need a Chow
form of the positive-dimensional algebraic set $\V$ with good height
properties.
\begin{Lemma}
  The algebraic set $\V$ admits a Chow form $\ChoV$ in
  $\Z[\U^0,\dots,\U^m]$ with $\ell_{v_\infty}(\ChoV) \le h_{\V} +
  (m+1)d_{\V}\log(m+n+2)$.
\end{Lemma}
\begin{proof}
  Let $\ChoV \in \Z[\U^0,\dots,\U^m]$ be a Chow form of $\V$ with
  integer coefficients and content~$1$. Let ${\sf M}_\qu =( \{ v_p,\,p
  \in {\cal P}\},\ \{v_\infty\})$ be the set of absolute values over
  $\Q$ introduced in Subsection~\ref{ssec:av}. Then, for every
  non-Archimedean valuation $v_p$ in ${\sf M}_\Q$,
  $\ell_{v_p}(\ChoV)=0$. The definition of the height of $\V$ implies
  that we have
$$h_{\V}={\sf m}(\ChoV,m+1,m+n+1)+ (m+1)d_{\V} \sum_{i=1}^{m+n} \frac{1}{2i}.$$
Using Inequalities ${\bf A_1}$ and ${\bf A_2}$ of
Subsection~\ref{ssec:avpol}, we conclude that $\ell_{v_\infty}(\ChoV)
\le h_{\V} + (m+1) d_{\V}\log(m+n+2)$.
\end{proof}

\noindent We can now conclude the proof of Proposition~\ref{Prop:DH1},
using the specialization property seen in the previous section. Let
$\ChoV \in \Z[\U^0,\dots,\U^m]$ be a Chow form of $\V$ as in the
previous lemma. Following Theorem~\ref{Prop:1}, we rewrite the
indeterminates $\U^0,\dots,\U^m$ of $\ChoV$ as
$$
\U_{(0)} = \left [ \begin{array}{c} U^0_{0} \\ \vdots \\ U^m_{0}
  \end{array} \right ],
\U_{(\Y)} = \left [ \begin{array}{ccc} U^0_{1} & \dots & U^0_{m} \\
    \vdots &&\vdots\\
    U^m_{1} & \dots & U^m_{m}\end{array} \right ],\
\U_{(\X)} = \left [ \begin{array}{ccc} U^0_{m+1} & \dots & U^0_{m+n} \\
    \vdots &&\vdots\\
    U^m_{m+1} & \dots & U^m_{m+n}\end{array} \right ];\ 
$$
then, we let $\ChoV_\varepsilon \in
\Z[\Y,\U,\U^1,\dots,\U^m,\varepsilon]$ be the polynomial obtained by
performing the following substitution in~$\ChoV$:
 $$
\U_{(0)} \leftarrow \left [ \begin{array}{c} U_{0} \\ Y_1 \\ \vdots \\ Y_m \end{array} \right ],\ \
\U_{(\Y)} \leftarrow \left [
  \begin{array}{ccc} 0 & \dots & 0 \\ -1 &
    \dots & 0 \\
    \vdots & \ddots &\vdots\\
    0 & \dots & -1\end{array} \right ],\ \
\U_{(\X)} \leftarrow \left [ \begin{array}{ccc} U_1 & \dots &
    U_n
    \\ \varepsilon U_{m+1}^1  & \dots & \varepsilon U_{m+n}^{1} \\
    \vdots &&\vdots\\
      \varepsilon U_{m+1}^m & \dots &  \varepsilon U_{m+n}^{m}\end{array} \right
    ];
$$
Theorem~\ref{Prop:1} shows that this polynomial is non-zero. Let
finally $\ChoV_0 \in \Z[\Y,\U,\U^1\dots,\U^m]$ be the coefficient of
lowest degree in $\varepsilon$ of $\ChoV_\varepsilon$. Then,
Theorem~\ref{Prop:1} shows that $\Cho$ divides $\ChoV_0$ in
$\Z[\Y,\U,\U^1,\dots,\U^m]$.

If we rewrite $\ChoV_0$ as a polynomial in variables $\U^1\dots,\U^m$
with coefficients in $\Z[\Y,\U]$, this implies that $\Cho$ divides one
of these coefficients, say $\ChoV_{0,0}$, in $\Z[\Y,\U]$, with
$\ChoV_{0,0} \ne 0$.

The polynomial $\ChoV_{0,0}$ is in $\Z[\Y,\U]$ and satisfies
$\ell_{v_\infty}(\ChoV_{0,0})\le h_{\V} + (m+1)d_{\V} \log(m+n+2)$,
since all its coefficients are coefficients of $\ChoV$.  Besides, it
has total degree at most $(m+1)d_{\V}$.  Since $\ChoV_{0,0}/\Cho$ has
integer coefficients, we deduce that
$\ell_{v_\infty}(\ChoV_{0,0}/\Cho) \ge 0$; then, ${\bf A_3}$ implies
that $\ell_{v_\infty}(\Cho) \le \ell_{v_\infty}(\ChoV_{0,0})+
4(m+1)d_{\V}\log(m+n+2)$, which yields
$$\ell_{v_\infty}(\Cho) \le h_{\V} + 5(m+1)d_{\V} \log(m+n+2).$$ 


\subsection{Valuation estimates}\label{ssec:val}

We prove the missing statements of Proposition~\ref{Prop:DH}. The
conclusion of the proof uses valuation estimates; the key lemma is the
following.

\begin{Lemma}
  For any non-Archimedean absolute value $v$ on $\Q(\Y)$, the inequalities
  $$\ell_v(a_n N_n) \le \ell_v(\Cho)
\quad\text{and}\quad
\ell_v(a_n^{\G_n} \num_n) \le \G_n \ell_v(\Cho)$$ hold.
\end{Lemma}
\begin{proof}
  We let here $\widetilde{\Cho} \in \Q(\Y)[\U]$ be the {\em monic}
  Chow form of $\VS$, so that the primitive Chow form $\Cho$ and its
  leading term $a_n$ satisfy $\Cho = a_n \widetilde{\Cho}$.  Lemma~5
  in~\citep{DaSc04} establishes the inequalities
$$ h_v(N_n) \le h_v(\widetilde{\Cho}) \quad \text{and}\quad 
h_v(\num_n) \le \G_n h_v(\widetilde{\Cho}),$$ with $h_v(f) =
\max(\ell_v(f),0)$ for any polynomial $f$. Since on one hand $\ell_v$
is always bounded from above by $h_v$, and since on the other hand
$h_v(\widetilde{\Cho})=\ell_v(\widetilde{\Cho})$ (because this
polynomial has a coefficient equal to 1), we deduce the alternative
form
$$ \ell_v(N_n) \le \ell_v(\widetilde{\Cho}) \quad \text{and}\quad 
\ell_v(\num_n) \le \G_n \ell_v(\widetilde{\Cho}).$$ Since $a_n
\widetilde{\Cho} = \Cho$, using ${\bf N_1}$, we deduce
$$\ell_v(a_n N_n)= \ell_v(a_n) + \ell_v(N_n) \le \ell_v(a_n) + \ell_v(\widetilde{\Cho} ) = \ell_v(\Cho)$$
and 
$$\ell_v(a_n^{\G_n} \num_n)= \ell_v(a_n^{\G_n}) + \ell_v(\num_n)
= \G_n \ell_v(a_n) + \ell_v(\num_n) \le \G_n
\ell_v(a_n) + \G_n \ell_v(\widetilde{\Cho}) = \G_n
\ell_v(\Cho).$$
Thus, all inequalities are proved.
\end{proof}

Let ${\sf V}$ be the set of all absolute values on $\Q(\Y)$ either of
the form $v_S$ for $S$ irreducible in $\Q[\Y]$, or of the form $v_p$,
for $p$ a prime.  By the discussion in Subsection~\ref{ssec:avpol},
for $f$ in $\Q(\Y)[\X]$, $v(f) \le 0$ holds for all $v$ in ${\sf V}$
if and only if $f$ is in $\Z[\Y][\X]=\Z[\Y,\X]$.

Since $\Cho$ is in $\Z[\Y,\X]$, we have that $\ell_v(\Cho) \le 0$ for
all $v$ in ${\sf V}$. By the previous lemma, we obtain
$$\ell_v(a_n N_n) \le 0
\quad\text{and}\quad \ell_v(a_n^{\G_n} \num_n) \le 0,$$ so that
$a_n N_n$ and $a_n^{\G_n} \num_n$ are in $\Z[\X,\Y]$. To
conclude the proof of Proposition~\ref{Prop:DH}, recall from
Proposition~\ref{Prop:DH1} that $\deg(\Cho,\Y) \le d_{\V}$, which can
be restated as $\ell_{v_{\deg{}}}(\Cho) \le d_{\V}$, where $v_{\deg{}}$
is the non-Archimedean absolute value introduced in
Subsection~\ref{ssec:av}; this implies $\G_n \ell_{v_{\deg{}}}(\Cho) \le
\G_n d_{\V}$.

Applying the former lemma to the absolute value $v_{\deg{}}$,
we thus prove the last two assertions of Proposition~\ref{Prop:DH},
finishing its proof.


\section{Proof of the main theorem}\label{sec:main}

We finally prove Theorem~\ref{theo:main} using interpolation
techniques. The results of~\citep{DaSc04} enable us to give height
bounds for specializations of $(N_1,\dots,N_n)$ and
$(\num_1,\dots,\num_n)$.  The results of the previous section then
make it possible to predict a denominator for the coefficients of
$(N_1,\dots,N_n)$ and $(\num_1,\dots,\num_n)$, so that polynomial
interpolation of the numerators is sufficient.

We focus only on $N_n$ and $\num_n$, since extending the results to
all $(N_1,\dots,N_n)$ and $(\num_1,\dots,\num_n)$ is straightforward.
All the notation introduced in the previous section is still in use in
this section. 


\subsection{Norm estimates for interpolation}

First, we give norm estimates for interpolation at integer points.
For any integer $M > 0$, we denote by $\Gamma_M$ the set of integers
$$\Gamma_M=\{1,\dots,M\}.$$
Let us fix $M$ and another integer $L \le M.$ We will use subsets of
$\Gamma_M^m$ of cardinality ${L}^m$ to perform evaluation and
interpolation.  To control the norm growth through interpolation at
these subsets in the multivariate case, the following ``univariate''
lemma will be useful. 
\begin{Lemma}\label{lemma:vdm}
  Let $\Lambda$ be a subset of $\Gamma_M$ of cardinality $L$ and let
  ${\bf V}$ be the $L\times L$ Vandermonde matrix built on $\Lambda$.
  Let ${\bf a}=(a_1,\dots,a_L)$ be in $\Q^{L}$, with
  $\ell_{v_\infty}(a_i) \le A$ for all $i$, and let
  $(b_1,\dots,b_L)={\bf V}^{-1} {\bf a}$. Then the inequality
  $\ell_{v_\infty}(b_i) \le A+L\log(M+1)+\log(L)$ holds for all $i$.
\end{Lemma}
\begin{proof}
  Let ${\bf W}$ be the inverse of ${\bf V}$. The upper bound given
  in~\cite[Eq.~(22.3)]{Higham02} shows that all entries $w_{i,j}$ of
  ${\bf W}$ satisfy $|w_{i,j}|=v_\infty(w_{i,j}) \le (M+1)^{L}$. Since
  all entries of ${\bf a}$ satisfy $v_\infty(a_i) \le e^A$, we deduce
  that all $b_i$ satisfy $v_\infty(b_i) \le L(M+1)^{L}e^A$. Taking
  logarithms finishes the proof.
\end{proof}

In the multivariate case, we rely on the notion of {\em
  equiprojectable} set~\citep{AuVa00}, which we recall here, adding a
few extra constraints to facilitate norm estimates later on.  Let us
define a sequence $\Lambda_1,\Lambda_2,\dots$ of subsets of $\Gamma_M,
\Gamma_M^2,\dots$ through the following process:
\begin{itemize}
\item $\Lambda_1$ is a subset of $\Gamma_M$ of cardinality $L$;
\item for $i\ge 1,$ assuming that $\Lambda_i$ has been defined, we take
 $\Lambda_{i+1}$ of the form
$$\Lambda_{i+1} = \cup_{\y \in \Lambda_i}\ \, \big ( \y \times \Lambda_{i,\y} \big),$$
where each $\Lambda_{i,\y}$ is a subset of $\Gamma_M$ of cardinality $L$.
\end{itemize}
Then, we say that $\Lambda \subset \Gamma_M^m$ is an $(M,L)$-{\em
  equiprojectable set} if it arises as the $m$th element $\Lambda_m$
of a sequence $\Lambda_1,\dots,\Lambda_m$ constructed as
above. Observe that such a set has cardinality ${L}^m$.

Let $\Q[\Y]_{L}$ be the subspace of $\Q[\Y]$ consisting of all
polynomials of degree less than $L$ in each variable $Y_1,\dots,Y_m$;
thus, $\Q[\Y]_{L}$ has dimension ${L}^m$. Associated to an
$(M,L$)-equiprojectable set $\Lambda \subset \Gamma_M^m$, we set up the
evaluation operator
$$\begin{array}{cccc}
ev_{\Lambda}:& \Q[\Y]_{L} & \mapsto & \Q^{{L}^m} \\
& f& \mapsto & [f (\y) ]_{\y \in \Lambda}.
\end{array}$$
We let  ${\bf M}_\Lambda$ be the matrix of this map, where
we use the canonical monomial basis for $\Q[\Y]_{L}$.
\begin{Prop}\label{prop:interp}
  The following holds:
  \begin{itemize}
  \item The map $ev_{\Lambda}$ is invertible.
  \item Let ${\bf a}$ be in $\Q^{L^m}$, with $\ell_{v_\infty}(a_i) \le
    A$ for each entry $a_i$ of ${\bf a}$, and let ${\bf b}={\bf
      M}_\Lambda^{-1} {\bf a}$. Then the inequality
    $\ell_{v_\infty}(b_i) \le A+mL\log(M+1)+m\log(L)$ holds for each
    entry $b_i$ of~${\bf b}$.
  \end{itemize}
\end{Prop}
\begin{proof}
  To evaluate a polynomial $f \in \Q[\Y]$ at $\Lambda$, we first see
  it as a polynomial in $\Q[Y_1][Y_2,\dots,Y_m]$ and evaluate all its
  coefficients at $\Lambda_1$. We obtain $L$ polynomials $\{f_\y,\ \y
  \in \Lambda_1\}$ in $\Q[Y_2,\dots,Y_m]$, and we proceed recursively
  to evaluate each $f_\y$. This implies that the matrix ${\bf
    M}_\Lambda$ factors as ${\bf M}_\Lambda={\bf M}_m \cdots {\bf
    M}_1$, where, up to permutation of the rows and columns, ${\bf M}_i$ is a
  block diagonal matrix, whose blocks are Vandermonde matrices of size
  $L$ built on the sets $\Lambda_{i,\y}$ (each of them being repeated
  ${L}^{m-i}$ times).
  
  Thus, $\bf M$ is invertible, of inverse $\bf M$ given by ${\bf
    M}_1^{-1} \cdots {\bf M}_m^{-1}$. The second point is now a
  direct consequence of Lemma~\ref{lemma:vdm}.
\end{proof}


\subsection{Good specializations}

We return to the study of an algebraic set $\V \subset \C^{m+n}$
satisfying Assumption~\ref{ass2}; we discuss here the ``good'' and
``bad'' specialization values for the polynomials $(T_1,\dots,T_n)$
and $(N_1,\dots,N_n)$.

For $\y=(y_1,\dots,y_m)$ in $\C^m$ and $F$ in $\Q(\Y)[\X]$, we denote
by $F_\y$ the specialized polynomial $F(\y,\X)$, assuming that the
denominator of no coefficient of $F$ vanishes at $\y$. We denote by
$\V_\y$ the fiber of the projection $\Pi_0: \C^{m+n}\to \C^m$
restricted to $\V$, that is, the algebraic set
$$\V_\y=\V \cap Z(Y_1 -y_1,\dots,Y_m-y_m) \ \subset \ \C^{m+n}.$$
Finally, we say that $\y$ is a {\em good specialization} if the
following holds:
\begin{itemize}
\item the denominator of no coefficient in $(T_1,\dots,T_n)$ vanishes
  at $\y$, so that all polynomials $T_{i,\y}=T_i(\y,\X)$ are
  well-defined;
\item the monic triangular set
  $(Y_1-y_1,\dots,Y_m-y_m,T_{1,\y},\dots,T_{n,\y})$ is the Gr\"obner
  basis of the defining ideal of $\V_\y$, for the lexicographic order
  $Y_1 < \cdots < Y_m < X_1 < \cdots < X_n$.
\end{itemize}
The following proposition shows that for any $L$, there exist
$(M,L)$-equiprojectable sets where all points are good
specializations, if we choose $M$ large enough.
\begin{Prop}\label{prop:spec}
  For any positive integer $L$, there exists an
  $(M,L)$-equiprojectable set $\Lambda$ such that all points in
  $\Lambda$ are good specializations, with $M=(3n
  d_{\V}+n^2)d_{\V}+L$.
\end{Prop}
\begin{proof}
  Theorem~2 in~\citep{Schost03} shows that there exists a non-zero
  polynomial $\Delta \in \Z[\Y]$ of degree at most $M_0=(3n
  d_{\V}+n^2)d_{\V}$ such that any $\y \in \C^m$ with $\Delta(\y)
  \ne 0$ is a good specialization; in what follows, we take $M=M_0+L$.

  We are going to use this to construct a sequence
  $\Lambda_1,\Lambda_2,\dots,\Lambda_m$ of $(M,L)$-equiprojectable
  sets in $\Gamma_M, \Gamma_M^2,\dots,\Gamma_M^m$, and we will take
  $\Lambda=\Lambda_m$. We will impose the following property for $i\le
  m$:
  \begin{itemize}
  \item [$({\bf P}_i)$] for all $\y=(y_1,\dots,y_i)$ in $\Lambda_i$, the
    polynomial $\Delta(y_1,\dots,y_i,Y_{i+1},\dots,Y_m)$ is not
    identically zero.
  \end{itemize}
The proof is by induction.
  \begin{itemize}
  \item For $i=1$, remark that there exist at most $M_0$ values $y_1$
    such that $\Delta(y_1,Y_{2},\dots,Y_m)$ vanishes identically, so
    that there exists a subset $\Lambda_1$ of $\Gamma_M$ of
    cardinality $L$ that satisfies~${\bf P}_1$.
  \item For $1 \le i<m$, assume that a subset $\Lambda_i$ satisfying
    ${\bf P}_i$ has been defined. Thus, for $\y=(y_1,\dots,y_i)$ in
    $\Lambda_i$, the polynomial
    $\Delta(y_1,\dots,y_i,Y_{i+1},\dots,Y_m)$ is not identically zero.
    Consequently, there exist at most $M_0$ values $y_{i+1}$ such that
    $\Delta(y_1,\dots,y_i,y_{i+1},Y_{i+2},\dots,Y_m)$ vanishes
    identically. Thus, there exists a subset $\Lambda_{i,\y}$ of
    $\Gamma_M$ of cardinality $L$ such that
    $\Delta(y_1,\dots,y_i,y_{i+1},Y_{i+2},\dots,Y_m)$ vanishes
    identically for no element $y_{i+1}$ in $\Lambda_{i,\y}$. Defining
    $\Lambda_{i+1} = \cup_{\y \in \Lambda_i}\ (\y \times
    \Lambda_{i,\y}),$ we see that this set satisfies ${\bf
      P}_{i+1}$.
\end{itemize}
  Taking $i=m$ shows that $\Lambda=\Lambda_m$  satisfies our requests.
\end{proof}


\subsection{Norm estimates at good specializations}

We continue with estimates for good specializations of the polynomials
$a_n N_n$ and $a_n^{\G_n} \num_n$. In addition to the constant $\G_n =
1 + 2{\sum}_{i \leq n-1} (d_i-1)$ defined in the previous section, we
will also use the following quantities:
\begin{eqnarray*}
\corr_{n} &=& 5\log(n+3){\sum}_{i \leq n} d_i \\
\I_{n} & = & \corr_{n} + 3\log(2){\sum}_{i \leq n-1} d_i(d_i-1).\label{eq:main}
\end{eqnarray*}
One verifies that these constants satisfy the following upper bounds:
\begin{eqnarray*}
\corr_{n} &\le& 5 \log(n+3) (d_\V+n)\\
\I_{n} & \le &  3 d_\V^2 + 5 \log(n+3) (d_\V+n).\label{eq:main2}
\end{eqnarray*}
Considering only the dependency in $(d_\V,h_\V)$, the following
proposition gives a bound linear in $(d_\V+h_\V)$ for the
specialization of $a_n N_n$; the bound for $a_n^{\G_n} \num_n$ is
quadratic.
\begin{Prop}\label{prop:normspec}
  Let $\y=(y_1,\dots,y_m) \in \Z^m$ be a good specialization, such
  that all entries $y_i$ satisfies $\ell_{v_\infty}(y_i) \le M$. Then,
  the polynomials $a_{n,\y} N_{n,\y}$ and $ a_{n,\y}^{\G_n} \num_{n,\y}$ are
  well-defined, and they satisfy
\begin{eqnarray*}
 \ell_{v_\infty}(a_{n,\y} N_{n,\y})  &\le& 2h_{\V} + (6m+5)d_{\V}\log(m+n+2) 
+ (m+1) d_{\V} \log(M) + m \log(d_\V+1)
\\
& & +\corr_n,\\[1mm]
 \ell_{v_\infty}(a_{n,\y}^{\G_n} \num_{n,\y})  &\le&  \G_n\big[2 h_{\V} + (6m+5)d_{\V}\log(m+n+2)
+(m+1) d_{\V} \log(M)\\
& & + m \log(d_\V+1)\big]+{\sf I}_n.
\end{eqnarray*}
\end{Prop}
\begin{proof}
  If $\y$ is a good specialization, then the monic triangular set
  $(T_{1,\y},\dots,T_{n,\y})$ is well-defined and generates a radical
  ideal. As a consequence, 
\begin{eqnarray*}
  D_{n} &=&\prod_{1 \leq i \leq n-1} \frac{\partial T_{i}}{\partial X_i} \mod \langle T_{1},\ldots,T_{n-1}\rangle\\
 \text{and}~~~ e_{n} &=& \prod_{1 \leq i \leq n-1} \res(\cdots \res( \frac{\partial T_{i}}{\partial X_i}, T_i, X_i),\cdots,T_1, X_1)
\end{eqnarray*}
can be specialized at $\y$. Since $N_n=D_n T_n \bmod \langle
T_1,\dots,T_{n-1} \rangle$ and $\num_n = e_{n} T_n$, this
establishes our first claim.

Let next $\widetilde{\ChoV}_{\y}$ be the monic Chow form of $\V_{\y}$
(this is a polynomial in $m+n+1$ variables) and let $d_{\V_{\y}}$ be
its degree. The height $h_{\V_\y}$ of $\V_\y$ is
$$h_{\V_{\y}}= \sum_{p \in {\cal P}} \ell_{v_p}(\widetilde{\ChoV}_{\y}) +
{\sf m}(\widetilde{\ChoV}_{\y},1,m+n+1)+ d_{\V_{\y}} \sum_{i=1} ^{m+n}
\frac{1}{2i}.$$ Since $\widetilde{\ChoV}_{\y}$ has a coefficient equal
to~$1$, for every non-Archimedean absolute value $v_p$ we have
$\ell_{v_p}(\widetilde{\ChoV}_{\y}) \ge 0$.  Thus, we get the
inequality
\begin{equation*}
  {\sf m}(\widetilde{\ChoV}_{\y},1,m+n+1) +
  d_{\V_{\y}} \sum_{i=1} ^{m+n} \frac{1}{2i} \ \le\ h_{\V_{\y}}.  
\end{equation*}
Let further $v_\y \subset\C^n$ be the 0-dimensional algebraic set
obtained by projecting $\V_\y$ on the $\X$-space, and let
$\widetilde{\Chov_\y}$ be its monic Chow form. Thus,
$\widetilde{\Chov}_\y$ is obtained by setting all variables
corresponding to $Y_1,\dots,Y_m$ to 0 in $\widetilde{\ChoV}_\y$.

Because $\y$ is a good specialization, applying Lemma~5
in~\citep{DaSc04} to $v_\y$ and $\widetilde{\Chov_\y}$ gives the
following upper bounds:
\begin{eqnarray*}
\ell_{v_\infty}(N_{n,\y}) &\le& {\sf m}(\widetilde{\Chov}_{\y},n+1,1) + \corr_n\\  
\ell_{v_\infty}(\num_{n,\y}) &\le& \G_n {\sf m}(\widetilde{\Chov}_{\y},n+1,1) + {\sf I}_n,
\end{eqnarray*}
which imply, since $a_{n,\y}$ is actually in $\Z$,
\begin{eqnarray*}
\ell_{v_\infty}(a_{n,\y} N_{n,\y}) &\le&\ell_{v_\infty}(a_{n,\y})+  
     {\sf m}(\widetilde{\Chov}_{\y},n+1,1) + \corr_n\\
\text{and}\quad
\ell_{v_\infty}(a_{n,\y}^{\G_n} \num_{n,\y}) &\le& 
\G_n \ell_{v_\infty}(a_{n,\y})+ \G_n {\sf m}(\widetilde{\Chov}_{\y},n+1,1) + {\sf I}_n.
\end{eqnarray*}
Because $\widetilde{\Chov}_{\y}$ is obtained by specializing 
indeterminates at 0 in $\widetilde{\ChoV}_{\y}$, we deduce as 
in~\citep{KrPaSo01} that ${\sf m}(\widetilde{\Chov}_{\y},n+1,1) \le
{\sf m}(\widetilde{\ChoV}_{\y},m+n+1,1)$.  Using inequality ${\bf
  A}_2$, we deduce further
\begin{eqnarray}
\ell_{v_\infty}(a_{n,\y} N_{n,\y}) &\le& 
\ell_{v_\infty}(a_{n,\y}) + {\sf m}(\widetilde{\ChoV}_{\y},1,m+n+1) +
d_{\V_\y} \sum_{i=1} ^{m+n} \frac{1}{2i} + \corr_n\notag \\
&\le&   \ell_{v_\infty}(a_{n,\y}) + h_{\V_\y} +\corr_n\label{eq:N}
\end{eqnarray}
and similarly
\begin{eqnarray}
\ell_{v_\infty}(a_{n,\y}^{\G_n} \num_{n,\y}) &\le& 
\G_n \ell_{v_\infty}(a_{n,\y}) + \G_n h_{\V_\y} +{\sf I}_n. \label{eq:T}
\end{eqnarray}
Next, we give upper bounds on $\ell_{v_\infty}(a_{n,\y})$ and on
$h_{\V_\y}$.
We start with $\ell_{v_\infty}(a_{n,\y})=\ell_{v_\infty}(a_n(\y))$.
Recall from Proposition~\ref{Prop:DH} that $a_n$ is a polynomial
with integer coefficients, of total degree bounded by $d_\V$ and with
$$\ell_{v_\infty}(a_n) \le h_{\V} + 5(m+1)d_{\V}\log(m+n+2).$$
Since all $y_i$ are integers of absolute value bounded by $M$, we
deduce that 
$$\ell_{v_\infty}(a_{n,\y}) \le\ell_{v_\infty}(a_n) + d_{\V} \log(M)+ m \log (1+d_{\V}).$$
The previous bound on $\ell_{v_\infty}(a_n)$ gives
$$\ell_{v_\infty}(a_{n,\y}) \le
h_{\V} + 5(m+1)d_{\V}\log(m+n+2) + d_{\V} \log(M)+ m \log
(1+d_{\V}).$$ 
Next, we need to control $h_{\V_\y}$, with
$$\V_\y=\V \cap Z(Y_1 -y_1,\dots,Y_m-y_m).$$
All polynomials $Y_i-y_i$ have degree 1 and satisfy
$\ell_{v_\infty}(Y_i-y_i) \le \log(M)$. Using the arithmetic B\'ezout
inequality given in Corollary~2.11 of~\citet{KrPaSo01}, we obtain the
upper bound
$$h_{\V_\y} \le  h_{\V} + m d_{\V} \log(M) + m d_{\V}\log(m+n+1).$$
Using the bounds on $\ell_{v_\infty}(a_{n,\y})$ and 
$h_{\V_\y}$, Equations~\eqref{eq:N} and~\eqref{eq:T} give our
result after a quick simplification.
\end{proof}


\subsection{Conclusion by interpolation}

Finally, we obtain the requested bounds on $N_n$ and $T_n$ using
interpolation at suitable equiprojectable sets. 

The degree bounds are already in~\citep{DaSc04}, and also follow from
Proposition~\ref{Prop:DH}. They state that, if we see $a_n N_n$ in
$\Z[\Y][\X]$, each coefficient of this polynomial is in
$\Z[\Y]_{L_1}$, with $L_1=d_{\V}+1$. Similarly, each coefficient of
$a_n^{\G_n} \num_n$ is in $\Z[\Y]_{L_2}$, with $L_2=\G_n
d_{\V}+1$. Let thus
$$M_1=(3n d_{\V}+n^2)d_{\V}+L_1\quad\text{and}\quad
M_2=(3n d_{\V}+n^2)d_{\V}+L_2.$$ For $i=1,2$, by
Proposition~\ref{prop:spec}, there exists an
$(M_i,L_i)$-equiprojectable set $\Lambda_i$ such that all points in
$\Lambda_i$ are good specializations. Hence, we will interpolate the
coefficients of $a_n N_n$ at $\Lambda_1$ and those of $a_n^{\G_n}
\num_n$ at~$\Lambda_2$, and deduce the height bounds on $N_n$ and
$T_n$ given in Theorem~\ref{theo:main}. As was said before, replacing
$\V$ by its projection $\V_\ell$ gives the analogue bounds for {\em
  all} polynomials $(N_1,\dots,N_n)$ and $(T_1,\dots,T_n)$.
 
\paragraph{Bound on $N_n$.}
Write $a_n N_n$ as
$$a_n N_n = \sum_{\i} g_{\i,n} X_1^{i_1} \cdots X_n^{i_n} + g_n X_n^{d_n},$$
where all multi-indices $\i=(i_1,\dots,i_n)$ satisfy $i_\ell < d_\ell$
for $\ell \le n$, and all coefficients $g_{\i,n}$ and $g_n$ are in
$\Z[\Y]$. Proposition~\ref{prop:normspec} shows that
for $\y$ in $\Lambda_1$, we have the inequality
$$ \ell_{v_\infty}(a_{n,\y} N_{n,\y})  \le 2h_{\V} + (6m+5)d_{\V}\log(m+n+2) + (m+1) d_{\V} \log(M_1)
+m\log(d_\V+1)+\corr_n.$$ Applying Proposition~\ref{prop:interp} to
interpolate each $g_{\i,n}$ and $g_n$, we deduce that they all satisfy
\begin{eqnarray*}
\ell_{v_\infty}(g_{\i,n}),\ \ell_{v_\infty}(g_n) &\le& 
2h_{\V} + (6m+5)d_{\V}\log(m+n+2) + (m+1) d_{\V} \log(M_1)  \\
&& + m\log(d_\V+1) +\corr_n +m L_1\log(M_1+1) + m \log(L_1).  
\end{eqnarray*}
To simplify this expression, we use the definition $L_1=d_\V+1$ and
the upper bounds
$$\corr_n \le 5 \log(n+3) (d_\V+n),
\quad
n+3 \le m+n+3,
\quad
m+n+2 \le m+n+3.$$ After a few simplifications, we
obtain that $\ell_{v_\infty}(g_{\i,n})$ and $\ell_{v_\infty}(g_n)$
both admit the upper bound
\begin{equation*}
2h_{\V} + 2m \log(d_\V+1) +
 \big((6m+10)d_{\V}+5n\big)\log(m+n+3) 
+ \big((2m+1) d_\V+m\big) \log(M_1+1).
\end{equation*}
We continue by remarking that we have the inequality
$$M_1+1 \le (m+n+3)^2 (d_\V+1)^2,$$
which gives
$$
\ell_{v_\infty}(g_{\i,n}),\ \ell_{v_\infty}(g_n) \le 
2 h_\V +\big((4 m +2)d_\V + 4m \big) \log(d_\V+1) + \big((10 m +  12)d_\V + 5n +2m\big)\log(m+n+3).
$$
Note that $\ell_{v_\infty}(a_n)$ satisfies the same upper bound, in
view of Proposition~\ref{Prop:DH}. To conclude, we write $N_n$ as
$$N_n = \sum_{\i} \frac{g_{\i,n}}{a_n}X_1^{i_1} \cdots X_n^{i_n} + 
\frac {g_n}{a_n} X_n^{d_n}.$$ After clearing common factors in the
coefficients $g_{\i,n}/{a_n}$ and ${g_n}/{a_n}$, the logarithmic
absolute value can increase by at most $4d_\V \log(m+1)$ (by ${\bf
  A_3}$), since we have seen that all polynomials involved have
degree at most $d_\V$. We let $\gamma_{\i,n}/\varphi_{\i,n}$ and
$\gamma_n/\varphi_n$ be the reduced forms $g_{\i,n}/{a_n}$ and
$g_n/{a_n}$, that is, obtained after clearing all common factors in
$\Z[\Y]$. This gives the height-related statement in the first point of
Theorem~\ref{theo:main}; the claim of the lcm of all $\varphi_{\i,n}$
and $\varphi_n$ follows, since this lcm divides $a_n$. 

\paragraph{Bound on $T_n$.} 
Similarly, for $\y$ in $\Lambda_2$, we have (from Proposition~\ref{prop:normspec})
$$ \ell_{v_\infty}(a_{n,\y}^{\G_n} \num_{n,\y})  \le  \G_n\big(2 h_{\V} + (6m+5)d_{\V}\log(m+n+2)
+(m+1) d_{\V} \log(M_2)+ m \log(d_\V+1)\big) +{\sf I}_n.$$
Proceeding for $a_n^{\G_n} \num_n$ as we did for $a_n N_n$, we
first write 
$$a_n^{\G_n} \num_n = \sum_{\i} b_{\i,n} X_1^{i_1} \cdots X_n^{i_n} + b_n X_n^{d_n},$$
where all multi-indices $\i=(i_1,\dots,i_n)$ satisfy $i_\ell <
d_\ell$ for $\ell \le n$, and all coefficients $b_{\i_n}$ and $b_n$ are
in $\Z[\Y]$. This time, we obtain after interpolation
\begin{eqnarray*}
\ell_{v_\infty}(b_{\i,n}),\ \ell_{v_\infty}(b_n)
&\le&
\G_n\big[2 h_{\V} +  (6m+5)d_{\V}\log(m+n+2) +  (m+1) d_{\V} \log(M_2) \\
&&   + m \log(d_\V+1)\big] + {\sf I}_n + m L_2\log(M_2+1) + m \log(L_2).
\end{eqnarray*}
Now, we use the upper bounds
$$\G_n \le 2d_\V,\quad
{\sf I}_n \le 3 d_\V^2 + 5 \log(m+n+3) (d_\V+n),\quad L_2 \le
2d_\V^2+1, \quad M_2+1 \le 2 (m+n+3)^2(d_\V+1)^2,$$ and $\log(L_2) \le
1+2\log(d_\V+1)$. We obtain the following upper bound on
$\ell_{v_\infty}(b_{\i,n})$ and $\ell_{v_\infty}(b_n)$:
$$\begin{array}{rcl}
  \ell_{v_\infty}(b_{\i,n}), \quad\ell_{v_\infty}(b_n) &\le &
4d_\V h_\V + 3 d_\V^2 + m + 2\big((4m+2)d_\V^2+md_\V+2m\big) \log(d_\V+1)\\[1mm]
&&+ \big((20 m+14)d_\V^2 + 5d_\V+ 5n+2m \big) \log(m+n+3).
\end{array}$$
To obtain bounds on $T_n$ itself, we recall that this polynomial is
monic in $X_n$; thus, it is enough to divide $a_n^{\G_n} \num_n$ by
its leading coefficient $b_n$ to recover $T_n$. As in the previous
case, clearing common factors may induce a growth in logarithmic
absolute value, this time by at most $4 \G_n d_\V \log(m+1) \le
8d_\V^2 \log(m+1)$ (since all polynomials involved have degree at most
$\G_n d_\V$ by Proposition~\ref{Prop:DH}). Taking this into account
gives the estimate
$$\begin{array}{rcl}
&& 4d_\V h_\V + 3 d_\V^2 + m + 2\big((4m+2)d_\V^2+md_\V+2m \big) \log(d_\V+1)\\[1mm]
&&+ \big((20 m+22)d_\V^2  + 5d_\V+ 5n+2m \big) \log(m+n+3).
\end{array}$$
The second point
in Theorem~\ref{theo:main} follows after a few quick simplifications.


\section{Application}\label{sec:example}

To conclude, we give details of an application of our results. We work
under our usual notation, and we suppose that we are given a system
$(f_1,\dots,f_n)$ in $\Z[\Y,\X]$, such that $\V=Z(f_1,\dots,f_n)$, and
such that the Jacobian determinant $J$ of $(f_1,\dots,f_n)$ with
respect to $(X_1,\dots,X_n)$ does not vanish identically on any
irreducible component of $\V$. As a consequence, $\V$ satisfies
the first condition of Assumption~\ref{ass2}; we will actually
suppose that $\V$ satisfies the second condition as well. 

These assumptions are satisfied if for instance $\V$ is the graph of a
dominant polynomial mapping $\C^n\to \C^n$, with $f_i$ of the form
$Y_i-\varphi_i(\X)$. More generally, we can make a few remarks on the
strength of these assumptions.
\begin{itemize}
\item If we did not make our assumption on the Jacobian determinant,
  it would still be possible to restrict the study to the components
  of $\V$ where $J$ does not vanish identically, by adjoining the
  polynomial $1-SJ$ to the system $(f_1,\dots,f_n)$, where $S$ is a
  new variable.
\item The second condition of Assumption~\ref{ass2} is stronger.  If
  we are not in a situation where we can guarantee it (as on the
  example above), the proper solution will be to replace the
  discussion below by a more general one that takes into account the
  {\em equiprojectable decomposition} of $\VS$~\citep{DaMoScWuXi05}. We
  do not consider this here.
\end{itemize}

Under our assumptions, the question we study here is the following. To
compute either $(T_1,\dots,T_n)$ or $(N_1,\dots,N_n)$, it is useful to
know in advance their degrees in the variables $\Y$ (exactly, not only
upper bounds, as in~\cite{DaSc04}): for instance, it can help
determine how far we proceed in a Newton-Hensel lifting process.

A natural solution is to use modular techniques, that is, to determine
the degrees after reduction modulo a prime $p$: indeed, for all $p$,
except a finite number, the degrees obtained by solving the system
modulo $p$ will coincide with those obtained over $\Q$. The obvious
question is then, how large to choose $p$ to ensure that this is
indeed the case, with a high enough probability? Before giving our
answer, we remark that in practice, one should as well reduce to the
case $m=1$ by restricting to a random line in the $\Y$-space; we will
not analyze this aspect, as the proof techniques are quite similar to
what we show here.

For a prime $p$, and a polynomial $f$ in $\Q(\Y)[\X]$, we denote by
$f_{p}$ the polynomial in $\F_p(\Y)[\X]$ obtained by reducing all
coefficients of $f$ modulo $p$, assuming the denominator of no
coefficient of $f$ vanishes modulo $p$. Besides, for $f$ in either
$\Q(\Y)[\X]$, or $\F_p(\Y)[\X]$, we let $\delta(f)$ be the maximum of
the quantities $\deg(a)+\deg(b)$, for any coefficient $a/b$ of $f$,
with $\gcd(a,b)=1$. Our question here will be to estimate
$(\delta(T_1),\dots,\delta(T_n))$.

The main result of this section is the following proposition; it takes
the form of a big-Oh estimate but the proof gives explicit
results. Remark that we choose to measure the size of $p$ using
quantities that can be read off on the system of generators
$(f_1,\dots,f_n)$, since this is the input usually available in
practice.
\begin{Prop}\label{Prop:appli}
  Suppose that $(f_1,\dots,f_{n})$ have height bounded by $h$ and
  degrees bounded by $d$. Then, there exists a non-zero integer $A$ of
  height
  $${\cal H}_A = O((m+n)^2 n^5 d^{4n+4} (n h + (m+n)^2 \log(d) )),$$
  such that, for any prime $p$, if $A \bmod p\ne 0$, the following
  holds:
  \begin{itemize}
  \item all polynomials $(T_{1,p},\dots,T_{n,p})$ are well-defined;
  \item the ideal $\langle T_{1,p},\dots,T_{n,p} \rangle$ is radical
    and coincides with the ideal $\langle f_{1,p},\dots,f_{n,p}
    \rangle$ in $\F_p(\Y)[\X]$;
  \item for all $\ell \le n$, the equality
    $\delta(T_\ell)=\delta(T_{\ell,p})$ holds.
  \end{itemize}
\end{Prop}
\noindent In other words, if $A \bmod p \ne 0$, by solving the system
$(f_{1,p},\dots,f_{n,p})$ in $\F_p(\Y)[\X]$ by means of the
polynomials $(T_{1,p},\dots,T_{n,p})$, we can read off the quantities
$(\delta(T_1),\dots,\delta(T_n))$. Note that the bound on ${\cal H}_A$
is polynomial in the B\'ezout number; we believe that this is hardly
avoidable (as long as we express it using $n,d,h$), though the
exponent $4$ may not be optimal.

The last part of this section will be devoted to prove this
proposition; first, we give an estimate on the random determination of
a ``good prime'' $p$, whose proof is a consequence
of~\cite[Th.~18.10(i)]{GaGe99}.
\begin{Prop}
  One can compute in time $(n \log(mdh))^{O(1)}$ an integer $p$
  such that $6 {\cal H}_A \le p \le 12{\cal H}_A$, and, with
  probability at least $1/2$, $p$ is prime and does not divide $A$.
\end{Prop}
Remark that for $p$ as above, arithmetic operations in $\F_p$ can be
done in $(n \log(mdh))^{O(1)}$ bit operations. For a concrete example,
suppose that $m=1$, $n=12$, that $f_1,\dots,f_{12}$ have height
bounded by $h=20$ and degrees bounded by $d=3$: this is already quite
a large example, since the B\'ezout number is $531441$. In this case,
evaluating explicitly all bounds involved in the former results shows
that we would compute modulo primes of about 124 bits: this is
routinely done in a system such as Magma~\citep{BoCaPl97}.

\medskip \noindent We now prove Proposition~\ref{Prop:appli}; we start
by constructing explicitly the integer $A$. For $\ell \le n$, recall
that we wrote in Theorem~\ref{theo:main}
$$T_\ell = \sum_{\i} \frac{\beta_{\i,\ell}}{\alpha_{\i,\ell}}X_1^{i_1} \cdots X_\ell^{i_\ell} +  X_\ell^{d_\ell}.$$
\begin{itemize}
\item For $\ell \le n$, we first let $A_{0,\ell}$ be any non-zero
  coefficient of one of the polynomials $\alpha_{\i,\ell}$, and take
  $A_0=A_{0,1}\cdots A_{0,n}$.
\item By assumption, the Zariski-closure of
  $\Pi_0(Z(f_1,\dots,f_n,J))$ is not dense, so it is contained in a
  hypersurface. Thus, there exists a non-zero polynomial $H \in
  \Z[\Y]$ such that $Z(H) \subset \C^m$ contains
  $\Pi_0(Z(f_1,\dots,f_n,J))$. We let $A_1$ be any non-zero
  coefficient of such a polynomial $H$.
\item Let $S$ be a new variable; then by construction, the ideal
  $\langle 1-SH, f_1,\dots,f_n,J\rangle \subset \Q[\Y,\X,S]$ is the
  trivial ideal $\langle 1\rangle$.  We let $A_2$ be a non-zero
  integer that belongs to the ideal generated by $(1-SH,
  f_1,\dots,f_n,J)$ in $\Z[\Y,\X,S]$.
\item For $\ell \le n$, let $g_\ell/h_\ell$ be a coefficient of
  $T_\ell$ that maximizes the sum
  $\deg(\beta_{\i,\ell})+\deg(\alpha_{\i,\ell})$. We let $A_{3,\ell}$
  be a non-zero integer such that if $A_{3,\ell} \bmod p \ne 0$,
  $g_\ell$ and $h_\ell$ remain coprime modulo $p$, and their degrees
  do not drop modulo $p$.  Theorem~7.5 in~\citep{GeCzLa92} shows the
  existence of a non-zero integer $a_{3,\ell}$ that satisfies the
  first requirement; we take for $A_{3,\ell}$ the product of
  $a_{3,\ell}$ by one coefficient of highest degree in $g_\ell$ and
  one in $h_\ell$. As before, we take $A_3=A_{3,1}\cdots A_{3,n}$.
\end{itemize}
We then let $A=A_0 A_1 A_2 A_3$, and we first show that this choice of
$A$ satisfies our requirements.
\begin{Lemma}
  For any prime $p$, if $A \bmod p\ne 0$, the conclusions
  of Proposition~\ref{Prop:appli} hold.
\end{Lemma}
\begin{proof}
  Let us fix $p$, and let us write $\VS=Z(f_{1},\dots,f_{n})
  \subset\overline {\Q(\Y)}^n$ and $\VS_p=Z(f_{1,p},\dots,f_{n,p})
  \subset\overline {\F_p(\Y)}^n$. If $A \bmod p \ne 0$, $A_0 \bmod p
  \ne 0$, so all polynomials $(T_{1,p},\dots,T_{n,p})$ are
  well-defined and still form a Gr{\"o}bner basis.  Since
  $(T_{1},\dots,T_{n})$ reduce $(f_{1},\dots,f_{n})$ to zero in
  $\Q(\Y)[\X]$, the reduction relation can be specialized modulo $p$,
  as it involves no new denominator.  We deduce that the zero-set
  $Z(T_{1,p},\dots,T_{n,p}) \subset \overline {\F_p(\Y)}^n$ is
  contained in $\VS_p$.

  If $A \bmod p \ne 0$, we also have $A_1 \bmod p\ne 0$ and $A_2 \bmod
  p \ne 0$; as a consequence, $H \bmod p \ne 0$ and thus the ideal
  $\langle f_{1,p},\dots,f_{n,p},J_p\rangle \subset \F_p(\Y)[\X]$ is
  the trivial ideal. This implies that $\VS_p$ is finite, by the
  Jacobian criterion, since the Jacobian determinant $J_p$ vanishes
  nowhere on $\VS_p$. Besides, we also obtain that the roots of
  $\langle f_{1,p},\dots,f_{n,p}\rangle$ have multiplicity 1; the
  claims in the previous paragraph show that this is the case for
  $\langle T_{1,p},\dots,T_{n,p}\rangle$ as well. Thus, to obtain the
  second point, it suffices to prove that $|Z(T_{1,p},\dots,T_{n,p})|=
  |\VS_p|$.
  
  First, we prove the inequality $|\VS_p| \le |\VS|$. Let
  $\y=(y_1,\dots,y_m) \in \overline{\F_p}^m$ be such that $J$ vanishes
  nowhere on the fiber $\V_\y=Z(f_{1,p}(\y,\X),\dots,f_{n,p}(\y,\X))
  \subset \overline{\F_p}^n$, and such that $|\V_\y|=|\VS_p|$; such an
  $\y$ exists, by~\cite[Prop.~1]{Heintz83} (therein, the author
  assumes that the extension $\overline{\F_p}(\Y)
  \to\overline{\F_p}(\Y)[\X]/\langle f_{1,p},\dots,f_{n,p}\rangle$ be
  separable: this is the case here by the Jacobian criterion). Let
  $\F_q$ be a finite extension of $\F_p$ that contains all coordinates
  of all points in $\V_\y$. Then, using Newton iteration modulo powers
  of $\langle p,Y_1-y_1,\dots,Y_m-y_m\rangle$, all points in $\V_\y$
  can be lifted to solutions of $(f_1,\dots,f_n)$ in $\Z_q[[\Y-\y]]$,
  where $\Z_q$ is a finite integral extension of $\Z_p$.  Since
  $\Z_q[[\Y-\y]]$ contains $\Z[\Y]$, the number of solutions of
  $(f_1,\dots,f_n)$ in an algebraic closure of the fraction field of
  $\Z_q[[\Y-\y]]$ is $|\VS|$. As a consequence, the cardinality of
  $\VS_p$, which equals that of $\V_\y$, is at most that of $\VS$, as
  claimed.

  Let $d_1,\dots,d_n$ be the degrees of $T_1,\dots,T_n$ in
  respectively $X_1,\dots,X_n$.  In view of the inclusion proved
  above, we deduce the inequalities
$$d_1 \cdots d_n \ =\ |Z(T_{1,p},\dots,T_{n,p})|
\ \le \ |\VS_p| \ \le\ |\VS|\ =\ d_1 \cdots d_n.$$
As said before, this establishes the second point of the proposition.

It remains to deal with the last point. If $A_3 \bmod p \ne 0$, then
for all $\ell \le n$, $A_{3,\ell} \bmod p \ne 0$; the definition we
adopted for $A_{3,\ell}$ ensures that in this case, $g_\ell$ and
$h_\ell$ remain coprime and keep the same degree modulo~$p$, as
needed.
\end{proof}

It remains to estimate $h(A)=h(A_0)+h(A_1)+h(A_2)+h(A_3)$. Combining
the results of the next paragraphs finishes the proof of
Proposition~\ref{Prop:appli}.

\paragraph{Height of $A_0$.}  By construction, using the notation of
Theorem~\ref{theo:main}, we have $h(A_0) \le {\cal H}'_1 + \cdots +
{\cal H}'_n$, with for all $\ell$
$${\cal H}'_\ell=O\big ( d^{2n}(n h  + mn\log(d)+ (m+n)\log(m+n))   \big).$$
In particular, we have
$$h(A_0) =O\big ( d^{2n}(n^2 h  + mn^2\log(d)+ n(m+n)\log(m+n))   \big).$$

\paragraph{Height of $A_1$.}
Next, we estimate the degree and height of the polynomial $H$.  Let
$\V'=Z(f_1,\dots,f_n,J)$; if $\V'$ is empty, we take $H=1$ and we are
done. Otherwise, we get $\dim(\V') \le m-1$. By B\'ezout's theorem,
the degree $d_{\V'}$ of $\V'$ is bounded from above by $nd^{n+1}$. Further,
note that $h(J) \le h'$, with $h'=n(h+\log(nd)+d\log(n+1))$, in view
of the discussion following~\cite[Lemma~1.2]{KrPaSo01}. Applying twice
the arithmetic B\'ezout theorem (in the form
of~\cite[Coro.~2.11]{KrPaSo01}), first to bound the height of $\V$ and
then of $\V'$, we deduce the inequality
\begin{equation*}
h_{\V'} \le nd^{n+1} (nh +h'+(m+2n+1)\log(m+n+1))
\end{equation*}
and thus
\begin{equation*}  \label{eq:HV}
h_{\V'} \le nd^{n+1} (2nh+n\log(nd)+nd\log(n+1)+(m+2n+1)\log(m+n+1)).
\end{equation*}
Let us decompose $\V'$ into its irreducible components
$\V'_1,\dots,\V'_K$. For each such component $\V'_k$, there exists a
subset $\Y_k$ of $Y_1,\dots,Y_m$ such that the Zariski-closure of the
projection of $\V'$ on the $\Y_k$-space is a hypersurface.

Fix $k \le K$, let $\varphi_k$ be the corresponding projection, and
let $\mathscr{W}_k$ be the Zariski-closure of $\varphi_k(\V'_k)$. The
degree of $\mathscr{W}_k$ is at most that of $\V'_k$;
by~\cite[Lemma~2.6]{KrPaSo01}, the height of $\mathscr{W}_k$ is at
most $h_{\V'_k} + 3md_{\V'_k} \log(n+m+1)$. As a consequence, using
the remarks on~\cite[p.~347]{Philippon95}, we deduce that there exists
a non-zero polynomial $H_k \in \Z[\Y_k]$ of degree at most $d_{\V'_k}$
and height at most $h_{\V'_k} + d_{\V'_k}(3m \log(n+m+1)+2)$ that
defines $\mathscr{W}_k$.

We can take $H=H_1\cdots H_K$. The degree of $H$ is bounded by $d_{\V'}$;
using~\cite[Lemma~1.2.1.b]{KrPaSo01}, we see that its height is bounded
by
$$h''=nd^{n+1} (2nh+(4m+2n+2)\log(m+n+1)+n\log(nd)+nd\log(n+1)+2).$$
We deduce in particular
$$h(A_1) = O(nd^{n+1} (nh+nd\log(n)+(m+n)\log(m+n))).$$

\paragraph{Height of $A_2$.} 
We are going to apply a suitable version of the arithmetic
Nullstellensatz~\cite[Th.~2]{KrPaSo01}. We would also like to mention
the recent work of~\citet{Jelonek05}: it gives finer estimates
for the degrees of polynomials in the effective Nullstellensatz, and
it might be possible to derive also better height estimates with his
technique. We will not pursue this here.

The polynomials $(1-SH,J,f_1,\dots,f_n)$ have degrees at most
$(d'',d',d,\dots,d)$, with $d'=nd$ and $d''=nd^{n+1}+1$; their heights
are bounded by $(h'',h',h,\dots,h)$, with $h'$ and $h''$ as above.
Definition~4.7 in~\citep{KrPaSo01} associates to such a system of
equations a {\em degree} $\delta$ and a {\em height} $\eta$; Theorem~2
in~\citep{KrPaSo01} then shows that we can take
\begin{equation}
  \label{eq:A2}
h(A_2) \le 
(m+n+2)^2 d'' (2\eta + (h''+\log(n+2))\delta + 21 (m+n+2)^2 d''\delta \log(d''+1)).  
\end{equation}
The quantities $\delta$ and $\eta$ satisfy the following inequalities.
Let $\Gamma$ be the set of $(n+2)\times(n+2)$ integer matrices with
coefficients of height at most $\nu=2(m+n+2)\log(d''+1)$. To a matrix
$\bA$ in $\Gamma$, associate the polynomials
$$g_{\bA,i}=A_{i,1} (1-SH) + A_{i,2} J + A_{i,3} f_1 + \cdots +
A_{i,n+2}f_n, \quad 1 \le i \le n+2$$ and the algebraic set
$\mathscr{W}_{\bf A} = Z(g_{\bA,1},\dots,g_{\bA,n+2})$.  Then $\delta
\le \max_{\bA \in \Gamma} d_{\mathscr{W}_{\bf A}}$ and $\eta \le
\max_{\bA \in \Gamma} h_{\mathscr{W}_{\bf A}}$.

To give good estimates on these quantities, we perform linear
combinations of the equations $g_{\bA,i}$ to partially triangulate
them, by eliminating $1-SH$ from all equations except the first one
and $J$ from all equations except the first two ones. As in the proof
of~\cite[Lemma~4.8]{KrPaSo01}, for any $\bA\in\Gamma$, the ideal
$\langle g_{\bA,1},\dots,g_{\bA,n+2} \rangle$ is equal to the ideal
$\langle g^\star_{\bA,1},\dots,g^\star_{\bA,n+2}\rangle$ with the
shape just described:
\begin{itemize}
\item $g^\star_{\bA,1}=A^\star_{1,1} (1-SH)+A^\star_{1,2} J + A^\star_{1,3} f_1 + \cdots +A^\star_{1,n+2}f_n$,
\item $g^\star_{\bA,2}=                     A^\star_{2,2} J + A^\star_{2,3} f_1 + \cdots +A^\star_{2,n+2}f_n$,
\item $g^\star_{\bA,i}=                     A^\star_{i,3} f_1 + \cdots +A^\star_{i,n+2}f_n$ for $i \ge 3$.
\end{itemize}
Besides, one can take the coefficients $A^\star_{1,j}$ as entries of
$\bA$, the coefficients $A^\star_{2,j}$ as minors of $\bA$ of size 2,
and the coefficients $A^\star_{i,j}$ as minors of $\bA$ of size 3, for
$i\ge 3$. This implies that we have
\begin{itemize}
\item $\deg(g^\star_{\bA,1}) \le d''$ and $h(g^\star_{\bA,1}) \le \ell''=h'' + \nu+ \log(n+2)$;
\item $\deg(g^\star_{\bA,2}) \le d'$ and $h(g^\star_{\bA,2}) \le \ell'=h' + 2\nu+\log(2) + \log(n+2)$;
\item $\deg(g^\star_{\bA,i}) \le d$ and $h(g^\star_{\bA,i}) \le \ell=h + 3\nu+\log(6) + \log(n+2)$ for $i \ge 3$.
\end{itemize}
It follows that $\delta \le d^n d' d''=n^2d^{2n+2}+nd^{n+1}$, as
showed in~\cite[Lemma~4.8]{KrPaSo01}. The estimate we obtain on $\eta$
is finer than the one in that lemma, though: we apply a first time the
arithmetic B\'ezout theorem~\cite[Coro.~2.11]{KrPaSo01}, to obtain a
bound on the height of $Z(g^\star_{\bA,3},\dots,g^\star_{\bA,n+2})$,
then intersect it with $Z(g^\star_{\bA,2})$ and
$Z(g^\star_{\bA,1})$. This results in the inequality
$$ \eta \le d^n\big(d'' d' ((n+2)\ell+(m+2n+3)\log(m+n+2)) + \ell' d''   + \ell'' d'\big),$$
from which a bound on $h(A_2)$ follows by means of~\eqref{eq:A2}.
All formulas given in this paragraph yield explicit bounds;
however, after a few simplifications, we find the big-Oh estimate
$$h(A_2) = O((m+n)^2 n^5 d^{4n+4} (n h + (m+n)^2 \log(d) )).$$

\paragraph{Height of $A_3$.} Let us fix $\ell \le n$. The corollary to
Theorem~7.5 in~\citep{GeCzLa92} shows that the integer $a_{3,\ell}$ has
height at most $2d^{2n}m ({\cal H}'_\ell + m \log(2d^{2n}+1))$; the
height bound of $A_{3,\ell}$ follows by adding $2{\cal H}'_\ell$.
This leads to an upper bound on $h(A_3)$ by summing for
$\ell=1,\dots,n$; we obtain
$$h(A_3)=O\big(d^{4n} (mn h+m^2n\log(d)+ m(m+n)\log(m+n))\big).$$

\bigskip
\paragraph{Acknowledgments}
We acknowledge the support of NSERC, the Canada Research Chairs
Program and MITACS; and of the Japanese Society for the Promotion
of Science (Global-COE program ``Maths-for-Industry'').


\bibliographystyle{plainnat} 
\bibliography{DaKaSc09}


\end{document}